\documentclass[conference,letterpaper]{IEEEtran}
\IEEEoverridecommandlockouts
\usepackage[utf8]{inputenc} 
\usepackage[T1]{fontenc}
\usepackage{url}              % provides \url{...}
\usepackage{cite}             % improves presentation of citations

\usepackage[cmex10]{amsmath}  % Use the [cmex10] option to ensure complicance
\usepackage{mleftright}       % fix to wrong spacing of \left-,
\mleftright                   % \middle- \right-commands 

\usepackage{graphicx}         % provides \includegraphics{...} to
\usepackage{booktabs} 
\usepackage[symbol]{footmisc}

\usepackage{algorithm} %Cemre's addition

\renewcommand{\thefootnote}{\fnsymbol{footnote}}
\newcommand{\authorlistandfunding}[1]{%
\begingroup%
\let\thefootnote\relax%
\footnotetext{#1}%
\endgroup%
}

\usepackage{amsmath,amsthm,amssymb}
\usepackage{enumerate}
\usepackage{color}
\usepackage{xcolor}
\usepackage{url}
\usepackage{balance}
\usepackage{graphicx} 
\usepackage{mathrsfs}
\usepackage{epsfig}
\usepackage{verbatim}
\usepackage{setspace}
\usepackage{cite}
\usepackage{bbm}

\usepackage{multicol}
\usepackage{lipsum}
\usepackage{etoolbox}
\usepackage{algpseudocode}
\usepackage{algorithm}
\usepackage{pgfplots}
  \pgfplotsset{compat=newest}
  \usetikzlibrary{plotmarks}
  \usetikzlibrary{arrows.meta}
  \usepgfplotslibrary{patchplots}
  \usepackage{grffile}
  \usepackage{amsmath}

\newcommand*{\rom}[1]{\expandafter\@slowromancap\romannumeral #1@}

\newtheorem{theorem}{Theorem}%[section]
\newtheorem{lemma}{Lemma}
\newtheorem{proposition}{Proposition}
\newtheorem{corollary}{Corollary}
\newtheorem{assumption}{Assumption}
\newtheorem{definition}{Definition}

\newtheorem{remark}{Remark}

\newcommand{\bbP}{\mathbb{P}}
\newcommand{\bbR}{\mathbb{R}}

\newcommand{\calD}{\mathcal{D}}

\newcommand{\calH}{\mathcal{H}}

\newcommand{\calS}{\mathcal{S}}

\newcommand{\calY}{\mathcal{Y}}

\newcommand{\hatcalH}{\hat{\calH}}

\newcommand{\removed}[1]{}


\newcommand{\mc}[1]{\mathcal{#1}}
\newcommand{\eps}{\varepsilon}

\newcommand{\R}{\mathbb{R}}

\newcommand{\br}[1]{\left(#1\right)}
\newcommand{\sbr}[1]{\left[#1\right]}
\newcommand{\cbr}[1]{\left\{#1\right\}}
\newcommand{\abs}[1]{\left|#1\right|}
\newcommand{\imag}[1]{\text{Image}\br{#1}}
\newcommand{\graph}[1]{\text{graph}\br{#1}}
\DeclareMathOperator*{\di}{\mathrm{d}\!}

\newcommand{\boldeps}{\boldsymbol{\eps}}
\newcommand{\bolddelta}{\boldsymbol{\delta}}

\newcommand{\regtradeoff}[1]{\mathcal{\uppercase{#1}}}
\newcommand{\regtradeoffindex}[2]{\regtradeoff{#1}_{#2}}

\newcommand{\regdp}[2]{\mc{R}\br{#1,#2}}
\newcommand{\regcomposedp}[3]{\mc{R}_{#3}\br{#1,#2}}

\newcommand{\regheterdp}[3]{\mc{C}_{#3}\br{#1,#2}}
\newcommand{\regheterdpsingleparam}[2]{\mc{C}_{#2}\br{#1}}

\usepackage{mathtools, stmaryrd}
\usepackage{xparse} 
\DeclarePairedDelimiterX{\Iintv}[1]{\llbracket}{\rrbracket}{\iintvargs{#1}}
\NewDocumentCommand{\iintvargs}{>{\SplitArgument{1}{,}}m}
{\iintvargsaux#1} %
\NewDocumentCommand{\iintvargsaux}{mm} {#1\mkern1.5mu,\mkern1.5mu#2}

\title{Composition Theorems for Multiple \\ Differential Privacy Constraints} 

\author{%
 \IEEEauthorblockN{Cemre Cadir, Salim Najib, and Yanina Y. Shkel}
                   }

\newcommand*{\extended}{}
\begin{document}
\maketitle

\authorlistandfunding{Alphabetical author order. Cemre Cadir, Salim Najib, and Yanina Y. Shkel are with the School of Computer \& Communication Sciences, École Polytechnique Fédérale de Lausanne (EPFL), Switzerland. Email: cemre.cadir@epfl.ch, salim.najib@epfl.ch, yanina.shkel@epfl.ch. This work is funded by the Swiss NSF grant number 211337.}

%%%%%%
%% Abstract: 
%% If your paper is eligible for the student paper award, please add
%% the comment "THIS PAPER IS ELIGIBLE FOR THE STUDENT PAPER
%% AWARD." as a first line in the abstract. 
%% For the final version of the accepted paper, please do not forget
%% to remove this comment!
%%

\begin{abstract}
The exact composition of mechanisms for which two differential privacy (DP) constraints hold simultaneously is studied. The resulting privacy region admits an exact representation as a mixture over compositions of mechanisms of heterogeneous DP guarantees, yielding a framework that naturally generalizes to the composition of mechanisms for which any number of DP constraints hold. This result is shown through a structural lemma for mixtures of binary hypothesis tests. Lastly, the developed methodology is applied to approximate $f$-DP composition.
\end{abstract}
%\begin{IEEEkeywords}
%differential privacy, composition theorem, heterogeneous composition, hypothesis testing, infimal convolution, mixture of hypothesis tests.
%\end{IEEEkeywords}
\section{Introduction}\label{sec:introduction}

In this work, we present new composition theorems for differentially private (DP) mechanisms \cite{dwork2006calibrating,dwork2006our} and apply them to approximate the composition of  $f$-differentially private ($f$-DP) mechanisms \cite{fdp}. Our work leverages the hypothesis testing perspective that has been shown to be especially useful for proving composition theorems for privacy preserving mechanisms \cite{wasserman2010statistical,compositiontheorem,dptv}.

\subsection{Differential Privacy and f-DP}
Differential privacy \cite{dwork2006calibrating,dwork2006our} is a widely studied worst-case privacy measure \cite{dwork2008differential,wasserman2010statistical,duchi2013local,kairouz2014extremal,bassily2015local,murtagh2015complexity,abadi2016deep,bun2016concentrated,compositiontheorem,mironov2017renyi,koskela2020computing,asoodeh2020better,liu2021machine,gopi2021numerical,koskela2021computing,zhu2022optimal,fdp,ldpsurvey,dptv,guerra2024composition}. It provides strong guarantees that limit the ability to distinguish between two neighboring databases: that is, two databases that differ in a single record.
\begin{definition}[Differential Privacy~\cite{dwork2006calibrating,dwork2006our}] \label{def:dp}
A mechanism $M : \calD \to \calY$ is $(\eps, \delta)$-differentially private (or $(\eps, \delta)$-DP), if for all pairs of neighboring databases $D_0, D_1 \in \calD$ and any subset $\calS \subseteq \calY$, we have
\begin{align}
\bbP[M(D_0) \in \calS] \leq e^\eps \bbP[M(D_1) \in \calS] + \delta \label{eq:dp_def}.
\end{align}
The special case with $\delta = 0$ is called pure $\eps$-DP.
\end{definition}

%In the literature, there are many variants of differential privacy \cite{bassily2015local,mironov2017renyi,fdp}. $f$-differential privacy \cite{fdp} is a variant motivated by hypothesis testing point of view of DP \cite{wasserman2010statistical,compositiontheorem}. 
Definition \ref{def:dp} could be alternatively formulated as a constraint on the trade-off function of the hypothesis test employed by the adversary \cite{wasserman2010statistical,compositiontheorem}. Motivated by this, \cite{fdp} introduced $f$-DP. In $f$-DP, trade-off functions $f$ have infinite degrees of freedom and can precisely describe the privacy of a mechanism. This is in contrast with $(\eps,\delta)$-DP where the privacy level is expressed with two parameters. One interesting feature of the $f$-DP framework is the equivalence of a symmetric trade-off function $f$ and a (potentially infinite) collection of $(\eps,\delta)$-DP guarantees. This equivalence suggests that an $f$-DP guarantee can be approximated by multiple DP constraints and motivates the study of such settings in this paper. By using multiple DP constraints, we can approach an arbitrary $f$-DP constraint and approximate its $k$-fold composition.

\subsection{Composition Theorems}

In most differentially private applications, the same database is reused for multiple queries as opposed to just a single query. %Formally, the \textit{composition} of mechanisms $M_1, \ldots, M_k$ is the mechanism $M = (M_1, \ldots, M_k)$. Composition is said to be adaptive when mechanisms are allowed to depend on the previous outputs and non-adaptive otherwise. 
Privacy degrades with composition and it is crucial to quantify the loss closely. The composition theorems characterize this privacy loss. Dwork et al. \cite{dwork2006our} provide an initial bound on $k$ compositions of $(\eps,\delta)$-DP and show that it is $(k\eps,k\delta)$-DP. Since then, there has been a large body of work on composition theorems for differential privacy \cite{dwork2006our,dwork2006calibrating,dwork2008differential,murtagh2015complexity,abadi2016deep,bun2016concentrated,compositiontheorem,fdp,koskela2020computing,asoodeh2020better,gopi2021numerical,koskela2021computing,zhu2022optimal,dptv,guerra2024composition}. Some recent works have taken advantage of the hypothesis testing point of view of DP \cite{wasserman2010statistical, compositiontheorem,murtagh2015complexity,dptv}. Kairouz et al. \cite{compositiontheorem} compute the exact $k$-fold composition region for a single $(\eps,\delta)$-DP constraint in the adaptive setting. Ghazi et al. \cite{dptv} extend this result by taking total variation privacy into account. This is equivalent to having simultaneous $(\eps,\delta)$-DP and $(0,\eta)$-DP constraints on all $k$ mechanisms. This approach achieves tighter composition results compared to \cite{compositiontheorem} whenever the mechanisms attain a non-trivial total variation constraint. Murtagh et al. \cite{murtagh2015complexity} prove a tight heterogeneous composition theorem in a similar manner. Dong et al. \cite{fdp} take a different approach and show a central-limit like composition result for sufficiently smooth trade-off functions, yielding an approximation method for composition.

\subsection{Our Contributions}

We build on the existing composition results and investigate the composition of double-DP constraint mechanisms, which satisfy $(\eps_1,\delta_1)$-DP and $(\eps_2,\delta_2)$-DP, simultaneously. In Theorems \ref{thm:main} and \ref{thm:alt-main}, we establish the exact composition region in this setting. Our result in Theorem \ref{thm:main} reveals the intimate relationship between the composition of double-DP constraint mechanisms and the heterogeneous composition. 
%It turns out that the composition of double-DP constrained mechanisms can be written as a mixture of compositions of heterogeneous single-DP constrained mechanisms. This relationship is a consequence of the hypothesis testing view of DP  together with Lemma \ref{lem:mixture} and Corollary \ref{corollary:mixture}. This link is particularly interesting as it can be naturally generalized to $n$-DP constrained mechanisms, hinting at how to compute the composition of $n$-DP constrained mechanisms via the heterogeneous composition.
Thus, we find the exact heterogeneous composition of $n$ $\eps_1$-DP mechanisms and $m$ $\eps_2$-DP mechanisms in Theorem \ref{thm:heter}. This result achieves a lower computational complexity compared to \cite{murtagh2015complexity}. In addition to facilitating Theorem \ref{thm:main}, Theorem \ref{thm:heter} is of independent interest. Its setting reflects a realistic scenario in which the privacy constraints need to be adjusted at some point to satisfy updated requirements such as the number of queries or the privacy budget.
Finally, in Propositions \ref{prop:approx_below} and \ref{prop:approx_above}, we develop a method to approximate $f$-DP composition via double-DP constraint composition, motivated by \cite{fdp}.

The remainder of this paper is structured as follows. In Section \ref{sec:prelim}, we introduce a basic lemma for the mixture of hypothesis tests. In Section \ref{sec:main:results}, we present our main results, namely the heterogeneous composition and double-DP constraint composition. We detail an approximation method for $f$-DP composition and illustrate it with an example in Section \ref{sec:fdp}. Section \ref{sec:conclusion} concludes the paper.

\section{Preliminaries}\label{sec:prelim}
\subsection{Hypothesis Testing and $f$-DP}
Binary hypothesis testing provides a very interesting point of view to understand and study DP. First, recall the binary hypothesis testing problem $\calH(P_0,P_1)$. Given probability distributions $P_0$ and $P_1$ on the same space:
\begin{align}
    \calH_0 &: Y \sim P_0  \nonumber \\
    \calH_1 &: Y \sim P_1. \nonumber%\label{eq:binHT}
\end{align}
The goal is to design a \textit{measurable} and \textit{possibly non-deterministic} decision rule $\hatcalH : \calY \to \{ 0, 1 \}$. For a given $\hatcalH$, we define two errors:
\begin{align}
      \beta_\mathrm{I} := \beta_\mathrm{I}(\hat\calH) &= \bbP [ \hatcalH = 1 | \calH = 0 ] \label{eq:typeI} \mbox{ and} \\
      \beta_\mathrm{II} := \beta_\mathrm{II}(\hat\calH) &= \bbP [ \hatcalH= 0 | \calH = 1 ], \label{eq:typeII} 
\end{align}
where (\ref{eq:typeI}) is type I error or the probability of false alarm, and (\ref{eq:typeII}) is type II error or the probability of missed detection. In the conventional hypothesis testing problem, the goal is minimizing these errors. However, from the perspective of privacy, we would like the adversary to have high errors when performing certain hypothesis tests. 
%Thus, we are interested in understanding the achievable errors in the adversarial hypothesis test for a given privacy guarantee. For this purpose, we use trade-off functions.

Given a hypothesis test $\calH(P_0,P_1)$, we define the trade-off function $f(P_0, P_1) : [0,1] \to [0,1]$ as
\begin{equation*}
f(P_0,P_1)(t) = \inf_{\hat\calH} \cbr{\beta_\mathrm{II} | \beta_\mathrm{I} \leq t}. %\label{eq:tradeoff}
\end{equation*}
We can interpret DP from the point of view of hypothesis testing \cite{wasserman2010statistical,compositiontheorem}. Define the trade-off function 
\begin{equation*}
    f_{\eps,\delta}(t) = \max \cbr{0, 1 - \delta - e^{\eps}t, e^{-\eps}( 1 - \delta - t)}, \ t \in [0,1].
\end{equation*}
An $(\eps,\delta)$-DP mechanism $M$ satisfies, for all neighboring datasets $D_0, D_1 \in \calD$, with $M(D_i) \sim P_{M(D_i)}$,
\begin{equation}
    f\br{P_{M(D_0)},P_{M(D_1)}} \geq f_{\eps,\delta} \label{eq:f_eps,delta}.
\end{equation}
% \begin{lemma}[$f$-DP generalizes $(\eps,\delta)$-DP,\cite{wasserman2010statistical,compositiontheorem,fdp}]\label{lem:ldp_hypotest}$M$ is $(\eps,\delta)$-DP if and only if $M$ is $f_{\eps,\delta}$-DP with
% \begin{align*}
%     f_{\eps,\delta}(t) = \max \cbr{0, 1- \delta - e^{\eps}t, e^{-\eps}( 1- \delta - t)}, \ t \in [0,1] \label{eq:ldp_trade}. \\
%     f_(\bbP[ \calM(\cdot) | D_0],\bbP[ \calM(\cdot) | D_1]) \leq f_{\eps,\delta}(t) 
% \end{align*}
% \end{lemma}

This interpretation inspired a very general variant of DP, called $f$-DP \cite{fdp}. Informally, a mechanism $M$ satisfies $f = f(P,Q)$-DP if differentiating two neighboring databases $D_0$ and $D_1$ through $M(D_0)$ and $M(D_1)$ is at least as difficult as differentiating $P$ and $Q$, in a strong sense. 
%Formally, letting $M = P_{Y|D}$ and denoting $P_{Y|D_i}$ the distribution of the output $M(D_i)$ for $i \in \cbr{0,1}$, $M$ is $f$-DP if and only if $f(P_{Y|D_0}, P_{Y|D_1}) \geq f$ for all neighboring $D_0, D_1$.
Equation \eqref{eq:f_eps,delta} indicates that $f$-DP reduces to $(\eps,\delta)$-DP with the appropriate trade of function, i.e. $f = f_{\eps,\delta}$. Likewise, multiple $(\eps_i,\delta_i)$-DP constraints could also be reflected in the $f$-DP framework: let $\boldeps = (\eps_1, \ldots, \eps_n) \in (\R^+)^n$, $\bolddelta = (\delta_1, \ldots, \delta_n) \in [0,1]^n$, and %for $0 \leq t \leq 1$,
\begin{equation}
    f_{\boldeps,\bolddelta}(t) = \max_{(\eps, \delta) \in (\boldeps,\bolddelta)} f_{\eps,\delta}(t) \ \forall t \in [0,1].
\end{equation}

\subsection{Privacy Regions}

Privacy regions offer an alternative way to express trade-off functions. They include all achievable error pairs and have one to one correspondence to the trade-off functions.
\begin{definition}[Privacy Region]
Let $f = f(P_0,P_1)$ be a trade-off function for a hypothesis test $\calH(P_0,P_1)$. The corresponding privacy region $\regtradeoff{f} \subseteq [0,1]^2$ is defined as the set of achievable $(\beta_\mathrm{I}, \beta_\mathrm{II})$ error pairs, as follows:
\begin{equation*}
    \regtradeoff{f} = \cbr{\begin{bmatrix}
        \beta_\mathrm{I} \\ \beta_\mathrm{II}
    \end{bmatrix} \in [0,1]^2 \ | \ 1 - \beta_\mathrm{I} \geq  \beta_\mathrm{II} \geq f(\beta_\mathrm{I})}. %\label{eq:priv-region}
\end{equation*}
\end{definition}

In our work, we state our results as bounds on privacy regions. As in \cite{compositiontheorem}, we write the privacy region associated with $f_{\eps,\delta}$ as $\regdp{\eps}{\delta}$.

We denote the privacy region of the intersection of the $n$ $(\eps_i,\delta_i)$-DP guarantees by 
\begin{equation}
    \regdp{\boldeps}{\bolddelta} = \bigcap_{i=1}^n \regdp{\eps_i}{\delta_i}.
\end{equation}
%with associated trade-off function $f_{\boldeps,\bolddelta} = \max_{i=1}^n f_{\eps_i,\delta_i}$. 
In other words, this intersection region signifies the privacy region of a mechanism with trade-off function $f_{\boldeps, \bolddelta}$.
We study the $k$-fold composition of such mechanisms.
\begin{definition}[Composition region for multiple DP constraints]  \label{def:mult_dp_comp}
With the same notation as above, we define the privacy region for the $k$-composition of $(\eps_i,\delta_i)$-DP mechanisms for $i\in\Iintv{1,n}:= [1,n] \cap \mathbb Z$ and $k \in \mathbb{N}^* := \mathbb{Z}_{\geq 1}$ to be $\regcomposedp{\boldeps}{\bolddelta}{k}$. 
\end{definition}

We also study the following heterogeneous composition.

\begin{definition}[Heterogeneous composition region] \label{def:het_comp}
Let $n,m \in \mathbb{N}$. The composition of $n$ $\eps_1$-DP and $m$ $\eps_2$-DP mechanisms is referred to as a heterogeneous composition, with privacy region $\regheterdp{\eps_1}{\eps_2}{n,m}$ or $\regheterdpsingleparam{\boldeps}{n,m}$ with $\boldeps = (\eps_1,\eps_2)$.
\end{definition}    
\subsection{Mixture of Hypothesis Tests}\label{sec:mixture}
In this section, we prove a new lemma on the mixture of hypothesis tests \cite{torgersen1977mixture,lecam1986asymptotic}. This lemma is at the core of our results, linking multiple DP constraint composition and heterogeneous composition. Let $n \in \mathbb{N}^*$ and $P^i,Q^i$ be pairs of distributions on $\mathcal{Y}$, for $i \in \Iintv{1,n}$. We denote the associated hypothesis tests by $\calH^i = \calH(P^i, Q^i)$. The mixture distributions are defined as \[
P(y,i) = \alpha_i P_{Y|I}(y|i) = \alpha_i P^i(y), \ Q(y,i) = \alpha_i Q^i(y)
\]
where $\boldsymbol{\alpha} = (\alpha_i)_{i=1}^n \in [0,1]^n$ are fixed and satisfy $\sum_{i=1}^n \alpha_i = 1$. Then, the observed-class mixture hypothesis test is $\calH_m = \calH(P, Q)$, where the test $\calH^i$ is selected with probability $\alpha_i$. Note that, once randomly selected, $i$ is known to the observer.

%Let $f_i = f(P^i,Q^i)$ be its corresponding hypothesis testing trade-off function, where the $i^{\text{th}}$ hypothesis test is $\calH^i = \calH(P^i, Q^i)$ given data $Y$. Fix $\boldsymbol{\alpha} = (\alpha_i)_{i=1}^n \in [0,1]^n$ such that $\sum_{i=1}^n \alpha_i = 1$. First, we define the mixture distributions $P(y,i) = \alpha_i P_{Y|I}(y|i) = \alpha_i P^i(y)$ and similarly, $Q(y,i) = \alpha_i Q^i(y)$. Then, the observed-class mixture hypothesis test is $\calH_m = \calH(P, Q)$, where the test $\calH^i$ is selected with probability $\alpha_i$.

\begin{lemma}%[Trade-off function of a mixture of hypothesis tests]
\label{lem:mixture}
Let $f_i = f(P^i,Q^i)$ be the trade-off function for $\calH^i$. The trade-off function $f_m = f(P,Q)$ of $\calH_m$ satisfies
\begin{equation}\label{equ:mixture}
    f_m(t) = \min_{\substack{t_i \in [0,1], \ i \in \Iintv{1,n}\\ \sum_{i=1}^n \alpha_i t_i = t}} \sum_{i=1}^n \alpha_i f_i(t_i).
\end{equation}
\end{lemma}
%Note that the $\min$ is meant in a strict sense: there exists a minimizer.
\begin{proof}
Consider that $f_m(t) = \inf_{\hat\calH_m}\cbr{\beta_\mathrm{II}^m|\beta_\mathrm{I}^m\leq t}$. For $\nu \in \cbr{\mathrm{I},\mathrm{II}}$, $\beta_\nu(\hat\calH_m) = \sum_{i=1}^n \alpha_i\beta_\nu^i(\hat\calH_m(\cdot,i)) := \sum_{i=1}^n \alpha_i\beta_\nu^i(\hat\calH^i)$ by routine arguments. The reformulated constraint becomes $\beta_\mathrm{I}^m = \sum_{i=1}^n \alpha_i t_i \leq t$, $t_i = \beta_\mathrm{I}^i(\hat\calH^i)$. Then observe that $f_i(t_i) = f_i(\beta_\mathrm{I}^i(\hat\calH^i)) \leq \beta_\mathrm{II}^i(\hat\calH^i)$, leading to the rewriting of $\beta_\mathrm{II}^m$ as $\sum_{i=1}^n \alpha_i f_i(t_i)$. The precise arguments can be found in \ifdefined\extended
Appendix \ref{subsec:mixture_lemma_proof}.
\else
   \cite{arxivextended}.
\fi
\end{proof} 
\begin{figure}[tb]
    \centering
\includegraphics[width=0.8\linewidth]{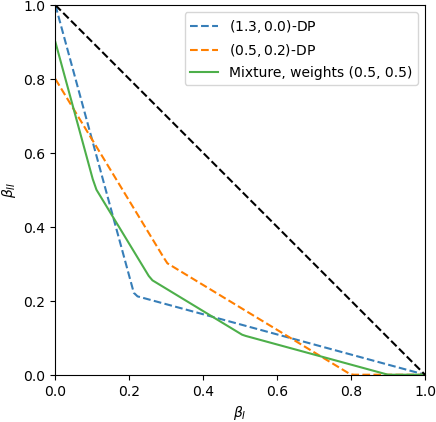}
    \caption{The trade-off functions $f_{1.3,0}$ and $f_{0.5,0.2}$, and their mixture with weights $\alpha_i = 0.5$. The mixture region does not necessarily include the union of the original regions. Likewise, it is not necessarily included in the union of the original regions.} 
    \label{fig:mixture}
\end{figure}
Equation (\ref{equ:mixture}) is known in the convex optimization literature as a weighted \textit{infimal convolution}. Its properties are studied in \cite{bauschke_convex_analysis}, from which the following corollary becomes apparent.
\begin{corollary}\label{corollary:mixture}
With the same notation as above,
\begin{equation}
    \regtradeoffindex{f}{m} = \cbr{\sum_{i=1}^n \alpha _i\begin{bmatrix}
        \beta_\mathrm{I}^i \\ \beta_\mathrm{II}^i
    \end{bmatrix} \ | \ \begin{bmatrix}
        \beta_\mathrm{I}^i \\ \beta_\mathrm{II}^i
    \end{bmatrix} \in \regtradeoffindex{f}{i} } := \sum_{i=1}^n \alpha_i\regtradeoffindex{f}{i}.
\end{equation}
\end{corollary}
In the cases that will be of interest to us, $f_i$s are piecewise affine functions. This makes the computation of $f_m$ and thus $\regtradeoffindex{f}{m}$ possible through properties of \textit{convex conjugates} of piecewise affine functions \cite{boyd_convex_optimization}. Details are in
\ifdefined\extended
Appendix \ref{subsec:corollary_mixture_proof} and Appendix \ref{subsec:computation_piecewise_affine}.
\else
   \cite{arxivextended}.
\fi

See Figure \ref{fig:mixture} for an example of the mixture of two equally weighted DP regions.

\section{Composition Theorems} \label{sec:main:results}
\subsection{Heterogeneous Composition} \label{sec:hetero}
%We first give the algorithm that computes privacy regions of heterogeneous compositions.
We state the theorem on heterogeneous composition (cf. Definition \ref{def:het_comp}) and the accompanying Algorithm \ref{alg:het_comp}.
\begin{algorithm}[tb]%[tb]
    \caption{Heterogeneous Composition
    }
    \label{alg:het_comp}
    \begin{algorithmic}[1]
        \State \textbf{Input:} $\eps_1, \eps_2, x, y$
        \State Find the set $S(\eps_1,\eps_2;x,y)$ of all $(a^*,b^*) \in \Iintv{0,x} \times \Iintv{0, y}$ such that $\eps_{a^*,b^*}^{x,y} = \eps_1(x-2a^*)+\eps_2(y-2b^*) \geq 0$.
        \State (Optional, for efficiency) Remove all but one $(a^*,b^*)$ that has the same slope $\eps_{a^*,b^*}^{x,y}$ from the set $S(\eps_1,\eps_2;x,y)$.
        \For{$(a^*,b^*)$ \textbf{in} $S(\eps_1,\eps_2;x,y)$} compute $\delta_{a^*,b^*}^{x,y}$ as
            \begin{align*}
              \delta_{a^*,b^*}^{x,y} =\  \br{\frac{1}{e^{\eps_1}+1}}^x \br{\frac{1}{e^{\eps_2}+1}}^y \sum_{b=0}^y\sum_{a=a_0(b)}^x \binom{x}{a}\binom{y}{b}\cdot\\\br{e^{a\eps_1 + b\eps_2}- e^{\eps_1(2(x-a^*)-a) + \eps_2(2(y-b^*)-b)}}
            \end{align*}
            where
            $
                a_0(b) = \max\br{0,\lceil (y-b^*-b)\frac{\eps_2}{\eps_1} + (x-a^*) \rceil}.
            $
        \EndFor
        \State Intersect the regions:
        \begin{align*}
        \regheterdp{\eps_1}{\eps_2}{x,y} = \bigcap_{(a^*,b^*) \in S(\eps_1,\eps_2;x,y)} \regdp{\eps_{a^*,b^*}^{x,y}}{\delta_{a^*,b^*}^{x,y}}.
    \end{align*}
        \State \textbf{Output:} $\regheterdp{\eps_1}{\eps_2}{x,y}$
    \end{algorithmic}
\end{algorithm}

\begin{theorem}%[Privacy region of the composition of heterogeneous mechanisms]
\label{thm:heter}

Let $\eps_1 > \eps_2 > 0$, and $x,y \in \mathbb{N}$. The exact privacy region of the composition of $x$ $(\eps_1,0)$-DP and $y$ $(\eps_2,0)$-DP mechanisms $\regheterdp{\eps_1}{\eps_2}{x,y}$ is computed by Algorithm \ref{alg:het_comp}.
\end{theorem}
\begin{proof}
Define distributions $P_i^1$ on $\cbr{0,3}$ and $P_i^2$ on $\cbr{1,2}$: 
\begin{equation*}
     P_i^1(x) = \begin{cases}
        \frac{e^{\eps_1}}{e^{\eps_1} + 1} &\text{ if } (i = 0, x=0) \text{ or }(i=1, x= 3),\\
        \frac{1}{e^{\eps_1} + 1} &\text{ if } (i = 0, x=3) \text{ or }(i=1, x= 0),
    \end{cases}
\end{equation*}
\vspace{1pt}
\begin{equation*}
    P_i^2(x) = \begin{cases}
        \frac{e^{\eps_2}}{e^{\eps_2} + 1} &\text{ if } (i = 0, x=1) \text{ or }(i=1, x= 2),\\
        \frac{1}{e^{\eps_2} + 1} &\text{ if } (i = 0, x=2) \text{ or }(i=1, x= 1).
    \end{cases}
\end{equation*}
For $j \in \cbr{1,2}$, let the mechanism $M^j$ output $X_0^j \sim P_0^j$ in the case of the null hypothesis, $X_1^j \sim P_1^j$ otherwise. Thus $M^j$ is binary randomized response \cite{dwork2006calibrating}, which achieves the whole of the $\eps_j$-DP region \cite{compositiontheorem}. Let $M^{x,y}$ be the composition of $x$ replicas of $M^1$ concatenated with $y$ replicas of $M^2$. Then $M^{x,y}$ outputs $X_0 = ((X_0^1)^x,(X_0^2)^y) \sim \Tilde{P}_0 = (P_0^1)^x(P_0^2)^y$ in the case of the null hypothesis, and $X_1 = ((X_1^1)^x,(X_1^2)^y) \sim \Tilde{P}_1 = (P_1^1)^x(P_1^2)^y$ otherwise, where $\Tilde{P}_0$ and $\Tilde{P}_1$ are distributions on $\Iintv{0,3}^{x+y}$.
It remains to compute the privacy region of $M$ and to show that it is the largest composition region in this case. This is done in a similar fashion as to \cite{compositiontheorem,murtagh2015complexity}. 
\ifdefined\extended
The details are in Appendix \ref{subsec:proof_heter}.
\else
   Due to space constraints, we leave the details for \cite{arxivextended}.
\fi
\end{proof}
Unlike \cite{murtagh2015complexity}, which, for a given $\delta_i$, returns the corresponding $\eps_i$, Algorithm 1 yields a closed form of the heterogeneous composition region, giving all the $(\eps_i,\delta_i)$ defining it.

\subsection{Composition Theorem for Double-DP Constraints}
In this section, we study the double-DP composition. We present two theorems that achieve the same privacy region. Theorem \ref{thm:main} depends on the heterogeneous composition regions $\regheterdpsingleparam{\boldeps}{x,y}$, and can be generalized to multiple-DP composition. Theorem \ref{thm:alt-main} provides a closed form expression and does not depend on $\regheterdpsingleparam{\boldeps}{x,y}$.
\begin{assumption}\label{assumption:main_thm}
Throughout this section, we let $\boldeps = (\eps_1, \eps_2) \in (\R^+)^2$ and $\bolddelta = (\delta_1,\delta_2) \in [0,1)^2$. We assume that $\delta_1 < \delta_2 \text{ and } (1-\delta_1)(1+e^{\eps_2}) < (1-\delta_2)(1+e^{\eps_1})$, making both $(\eps_i,\delta_i)$-DP constraints active. We also let $k\in\mathbb{N}^*$.
\end{assumption}
\begin{remark}\label{remark:intersection_single_dp}
A first upper bound of the double-DP $k$-composition privacy region is the intersection of the privacy regions for each DP guarantee, i.e. $
\regcomposedp{\boldeps}{\bolddelta}{k} \subseteq \regcomposedp{\eps_1}{\delta_1}{k} \cap \regcomposedp{\eps_2}{\delta_2}{k}$, known thanks to \cite{compositiontheorem}.  
\end{remark}
Remark \ref{remark:intersection_single_dp} gives a relatively trivial baseline for our main composition result. This baseline can be improved by leveraging the total variation bound induced by DP, as follows:
\begin{remark}\label{remark:intersection_dptv}
As discussed in \cite{dptv}: Let $\eta = \delta_2 + (1-\delta_2)\frac{e^\eps_2-1}{e^\eps_2+1}$. Then, $\regcomposedp{\boldeps}{\bolddelta}{k} \subseteq \regcomposedp{\Tilde{\boldeps}}{\Tilde{\bolddelta}}{k} \cap \regcomposedp{\eps_2}{\delta_2}{k}$, where $\Tilde{\boldeps} = (\eps_1, 0)$ and $\Tilde{\bolddelta} = (\delta_1,\eta)$.
\end{remark}
\noindent In our main composition results, we derive the exact double-DP composition privacy region which improves on Remarks 1 and 2.
%As seen below, these bounds are not tight in general.
\begin{theorem}%[Privacy region of the composition of double-DP constrained mechanisms]
\label{thm:main}
%Let $\eps_1, \eps_2 > 0$ and $\delta_1,\delta_2 \in [0,1]$ such that $\delta_1 < \delta_2 \text{ and } (1-\delta_1)(1+e^{\eps_2}) < (1-\delta_2)(1+e^{\eps_1})$. 
The privacy region of the composition of $k$ mechanisms which are $(\eps_1,\delta_1)$ and $(\eps_2,\delta_2)$-DP is
\begin{align*}
\regcomposedp{\boldeps}{\bolddelta}{k} = \Tilde{\delta} \regdp{0}{1} + (1-\Tilde{\delta}) \sum_{i=0}^k\binom{k}{i}(1-\alpha)^i \alpha^{k-i}\regheterdpsingleparam{\boldeps}{i,k-i}
\end{align*}
where $\boldeps = (\eps_1,\eps_2)$, $\bolddelta = (\delta_1,\delta_2)$, $\Tilde{\delta} = (1 - (1-\delta_1)^k)$ and
\begin{align*}
    \alpha = \frac{(1-\delta_1)e^{\eps_2} - (1-\delta_2)e^{\eps_1} + (\delta_2-\delta_1)}{(e^{\eps_2} - e^{\eps_1})(1-\delta_1)}.
\end{align*}
\end{theorem}
\begin{proof}
We start with the case $\delta_1 = 0$, implying $\Tilde{\delta}= 0$. Reusing the notation of Theorem \ref{thm:heter}'s proof, define the distributions on $\Iintv{0,3}$, for $i\in \cbr{0,1}$
\begin{equation*}
    P_i(x) = \begin{cases}
            (1-\alpha)P_i^1(x) &\text{ if } x \in\cbr{0,3},\\
        \alpha P_i^2(x) &\text{ if } x \in\cbr{1,2}.
    \end{cases}
\end{equation*}
Let $M$ be the mechanism that outputs $X_i \sim P_i$ under the hypothesis $\calH_i$. It is easy to check that $M$ has privacy region $\regdp{\boldeps}{\bolddelta}$. Observe that $M$ picks $M^1$ with probability $1-\alpha$ and $M^2$ with probability $\alpha$. Let $M_k$ be the $k$-composition of replicas of $M$, thus $M_k$ picks $M^1$ $i$ times with probability $1-\alpha$ each time independently, thus with probability $\binom{k}{i}(1-\alpha)^i\alpha^{k-i}$ overall $M_k$ picks the mechanism $M^{i,k-i}$. Thus the result follows by applying Corollary \ref{corollary:mixture} and Theorem \ref{thm:heter}. We extend to $\delta_1 > 0$ in 
\ifdefined\extended
Appendix \ref{subsubsec:extension_delta_1}.
\else
 \cite{arxivextended}.
\fi
\end{proof}
%\noindent An alternative way to describe the same region that does not rely on $\regheterdpsingleparam{\boldeps}{i,k-i}$ is given in the following theorem.
\begin{theorem}%[Privacy region of the composition of double-DP constrained mechanisms, alt.] 
\label{thm:alt-main}
Under Assumption \ref{assumption:main_thm},
\begin{equation*}
\regcomposedp{\boldeps}{\bolddelta}{k} = \Tilde{\delta}\regdp{0}{1} + (1-\Tilde{\delta}) \bigcap_{\substack{u,v \in \Iintv{0,k} \\ u \geq \left\lceil \frac{k\eps_1 - v(\eps_1 -\eps_2)}{\eps_1 + \eps_2}\right\rceil}} \regdp{\eps_{u,v}}{\delta_{u,v}}
\end{equation*}
where $\eps_{u,v} = \eps_1(u+v-k) + \eps_2(u-v)$ and
\begin{align*}
\delta_{u,v} = \sum_{\substack{(a,b,c,d)\\ \in B(\boldeps;u,v)}} \binom{k}{a,b,c,d}\br{\frac{1-\alpha}{e^{\eps_1}+1}}^{a+d} \br{\frac{\alpha}{e^{\eps_2} + 1}}^{b+c}\cdot \\ \br{e^{a\eps_1 + b\eps_2} - e^{(d+u+v-k)\eps_1 + (c+u-v)\eps_2}},
\end{align*}
with $(a,b,c,d) \in B(\boldeps;u,v) \subseteq \Iintv{0,k}^4$ if and only if
\begin{align*}
\begin{cases}
    a+b+c+d = k,\\
    (a+k-d-u-v)\eps_1 + (b+v-c-u)\eps_2 > 0.
\end{cases}
\end{align*}
\end{theorem}
\begin{proof}
With the same notation as the proof sketch of Theorem \ref{thm:main}: instead of looking at $M_k$ as a mixture of mechanisms, we can also directly tackle it as follows. $M_k$ is the mechanism that will output $(X_i)_1^k \sim \Tilde{P}_i = P_i^k$ under hypothesis $\calH_i$. Then, proceeding again as in the proofs of the main theorems of \cite{compositiontheorem,dptv,murtagh2015complexity}, we find the privacy region of $M_k$ and show that it is the largest that can be. Details are in \ifdefined\extended
Appendix \ref{subsec:proof_alt_main}.
\else
 \cite{arxivextended}.
\fi\end{proof}
%\textbf{Note that this method is trickier - and the results are exactly the same. }
\begin{figure}[tb]
    \centering
    \includegraphics[width=0.8\linewidth]{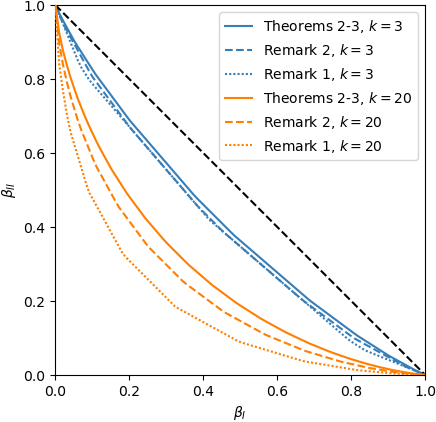}
    \caption{The privacy region of $(\boldeps, \bolddelta)$-DP under $k$-fold composition with $\boldeps = (0.3, 0.15)$, $\bolddelta = (0, 0.02)$ and $k \in \{3, 20\}$, as computed according to our result (Theorems 2-3) and prior works in Remarks \ref{remark:intersection_single_dp} and \ref{remark:intersection_dptv}. It is apparent that the previous bounds are close to the exact privacy region in the high privacy regime and with small $k$. As $k$ increases, these approximations rapidly worsen. }%Comparing our result (Theorems 2-3) to prior work, with $\boldeps = (0.3, 0.15)$ and $\bolddelta = (0, 0.02)$.}
    \label{fig:main_theorem_comparison}
\end{figure}
See Figure \ref{fig:main_theorem_comparison} for comparison of the bounds given by Remarks \ref{remark:intersection_single_dp} and \ref{remark:intersection_dptv} with the exact privacy region computed in Theorems \ref{thm:main} and \ref{thm:alt-main}. %It is apparent that the previous bounds are close to the exact privacy region in the high privacy regime and with small $k$. As $k$ increases, these approximations rapidly worsen.
\section{A Method for Approximating $f$-DP Composition} \label{sec:fdp}
In the following, we algorithmically approximate the $k$-composition of a trade-off function $f$, yielding both lower and upper privacy bounds\footnote{The corresponding code is public on github at {https://go.epfl.ch/multiDP}.}. This is done by approximating $f$ itself from below by a double-DP trade-off function $f_{\boldeps^-,\bolddelta^-}$, and same from above by $f_{\boldeps^+,\bolddelta^+}$, then computing the $k$-composition of each using Theorems \ref{thm:main} or \ref{thm:alt-main}. The function $f_{\boldeps^-,\bolddelta^-}$ is defined by\begin{equation}
    f_{\boldeps^-, \bolddelta^-} \in\arg\min_{\substack{(\boldeps, \bolddelta) \in (\bbR^+)^2 \times [0,1]^2\\ 0 \leq  f_{\boldeps, \bolddelta} \leq f} } \int_0^1 f(t) - f_{\boldeps, \bolddelta}(t) \di t\label{equ:approx_below}
\end{equation} and, similarly, $f_{\boldeps^+,\bolddelta^+}$ is defined by\begin{equation}
    f_{\boldeps^+,\bolddelta^+} \in \arg\min_{\substack{(\boldeps, \bolddelta) \in (\bbR^+)^2 \times [0,1]^2\\ 1 \geq f_{\boldeps, \bolddelta} \geq f} } \int_0^1 f_{\boldeps, \bolddelta}(t) - f(t) \di t\label{equ:approx_above}.
\end{equation}

The $f$-DP framework \cite{fdp} provides central-limit like approximate composition upper and lower privacy bounds, which become tight asymptotically in the number of composed mechanisms. These privacy bounds are given in terms of $G_\mu = f(\mc{N}(0,1), \mc{N}(\mu,1))$ trade-off functions and require the computation of information-theoretic quantities about the composed trade-off functions. The analysis below circumvents the latter and provides guarantees in terms of more familiar $(\eps,\delta)$-DP regions.

\begin{assumption}
$f$ is strictly convex, twice differentiable, and its graph is symmetric w.r.t. the line $y=x$, i.e. $f^{-1} = f$. Let $c \in (0,1)$ be the unique fixed point of $f$, i.e. $f(c) = c$.\label{assumption:approx}
\end{assumption}
\subsection{$f$-DP approximation from below}
Below, we define the normal rotation of $f$. It is the rotation of the graph of $f$ by $\frac{\pi}{4}$, changing the axes of the coordinate system from $y = 0$ and $x = 0$ to $y = -x$ and $y = x$. This allows the removal of the self-symmetry constraint, turning it into an evenness constraint in the rotated domain, which is easier to deal with.
\begin{definition}[Normal rotation of a trade-off function]
Let $f: [0,1] \to [0,1]$ be a trade-off function. Let $z = -\frac{f(0)}{\sqrt{2}}$. For $u \in [z,0]$, denote by $x_u \in [0,c]$ the unique solution of $u = \frac{x-f(x)}{\sqrt{2}},$ where $c \in [0,1]$ is the unique fixed point of $f$. The normal rotation of $f$ is $g:[0,z] \to \R$ such that \[g(u) = \frac{x_u + f(x_u)}{\sqrt{2}}.\]
\end{definition}
For some trade-off functions $f$ it is possible to find $x_u$ analytically. Otherwise, a numerical solver can be used. It can also be shown that $g$ is as smooth as $f$, and $g'(u) = \frac{1+f'(x_u)}{1-f'(x_u)}$, $g''(u) = \frac{2\sqrt{2}f''(x_u)}{(1-f'(x_u))^3}$.\\

With the normal rotation in hand, we can find $f_{\boldeps^-,\bolddelta^-}$. 
\ifdefined\extended
The proof of the following proposition is in Appendix \ref{subsec:proof_approx_below}.
\else
 We defer the proof to \cite{arxivextended} due to space constraints.
\fi
\begin{proposition}%[Approximation from below]
\label{prop:approx_below}
Let $f$ satisfy Assumption \ref{assumption:approx}, and let $g$ be its normal rotation. \\
Denote by $t^* \in [z,0]$ a root of\begin{equation*}
    g\br{\frac{t+z}{2}} + g'\br{\frac{t+z}{2}}\frac{t-z}{2} = g\br{\frac{t}{2}}+g'\br{\frac{t}{2}}\frac{t}{2}
\end{equation*}
then let $t_1 = \frac{t^*+z}{2}$ and $t_2 = \frac{t^*}{2}$. \\
If $t^* \neq z$, $f_{\boldeps^-,\bolddelta^-}$ is given, for $i \in \cbr{1,2}$, by\[
\eps_i^- = \ln\br{-\frac{g'(t_i)-1}{g'(t_i) + 1}}, \ \delta_i^- = 1-\frac{\sqrt{2}\br{g(t_i)-g'(t_i)t_i}}{g'(t_i)+1}.
\]
If $t^* = z$, $f_{\boldeps^-,\bolddelta^-} = f_{\eps_2^-,\delta_2^-}$.
\end{proposition}
\subsection{$f$-DP approximation from above}
When approximating $f$ from above, it can be shown that the first affine piece must intersect $f$ at $(0,f(0))$ and the second at $(c,f(c)) = (c,c)$. It thus suffices to pick the angular point $t^* \in [0,c]$ optimally. The detailed proof is in \ifdefined\extended
Appendix \ref{subsec:proof_approx_above}.
\else
  \cite{arxivextended}.
\fi
\begin{proposition}%[Approximation from above]
\label{prop:approx_above}
Let $f$ satisfy Assumption \ref{assumption:approx}. Denote by $t^* \in (0,c)$ the only solution of $f'(t) = \frac{f(c)-f(0)}{c-0}.$ %Let $h(t) = tf(0) + c(f(t)-f(0)-t)$.
If $h(t^*) :=  t^*f(0) + c(f(t^*)-f(0)-t^*)\geq 0$, $f_{\boldeps^+,\bolddelta^+} = f_{\eps^+,\delta^+}$ with\[
\eps^+ = \ln\br{\frac{f(0)-c}{c}}, \ \delta^+ = 1-f(0).
\]
If $h(t^*) < 0$, $f_{\boldeps^+,\bolddelta^+}$ is given by
\[
\eps_1^+ = \ln\br{\frac{f(0)-f(t^*)}{t^*}}, \ \delta_1^+ = 1-f(0), 
\]
\[
\eps_2^+ = \ln\br{\frac{c-f(t^*)}{t^*-c}}, \ \delta_2^+ = 1-c\br{1+e^{\eps_2^+}}.
\]
\end{proposition}
%Again due to space constraints, we defer the proof to \cite{arxivextended}.

\subsection{Example: Gaussian Mechanism}
%We illustrate our approximation method on $G_\mu$-DP mechanisms as introduced in \cite{fdp}. 
\begin{figure}[tb]
    \centering
    \includegraphics[width=0.8\linewidth]{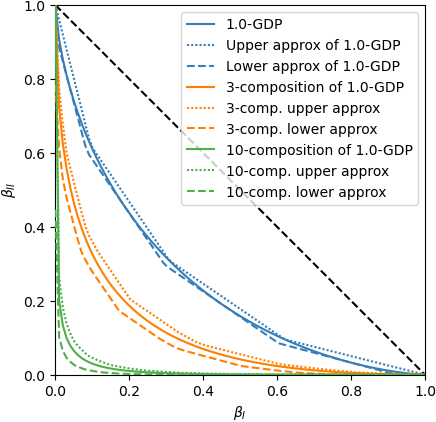}
    \caption{Privacy region of Gaussian mechanism with $\mu=1$ ($G_1$-DP) under $k$-fold composition for $k \in \{3,10\}$. The exact composition of $k$ $G_\mu$-DP mechanisms is $G_{\mu\sqrt{k}}$-DP \cite{fdp}. Lower and upper approximations are computed according to Propositions \ref{prop:approx_below} and \ref{prop:approx_above}, through double-DP composition in Theorems \ref{thm:main} and \ref{thm:alt-main}. Observe that the best approximation from below of $G_1$ is not a DP-TV trade-off function, improving on the approximation in \cite{dptv}.
    %Approximating $G_1$ and the $3$ \& $10$-compositions of $G_1$
    }
    \label{fig:gaussian_approx}
\end{figure}
In Figure \ref{fig:gaussian_approx}, we obtain empirically close lower and upper bounds on $1$-GDP compositions through double-DP, which will become arbitrarily tight once we generalize our results to $n$-DP composition, as $n$ increases.

\section{Concluding Remarks} 
\label{sec:conclusion}

In our work, we propose new composition theorems for heterogeneous composition and double-DP constraint composition. We illustrate the intimate relationship between these settings and leverage our results to approximate $f$-DP composition regions. 

Note that FFT based methods \cite{koskela2020computing,gopi2021numerical} can already achieve arbitrarily close bounds, but are harder to use and more computationally expensive. 

Building on the framework introduced in Theorem \ref{thm:main}, the generalization to $n$ DP constraints seems natural and will be the subject of future work. With $\boldsymbol{j} = (j_1, \ldots, j_n) \in (\mathbb{N}^*)^n$, we refer to the privacy region of the heterogeneous composition of $j_i$ $\eps_i$-DP for $i \in \Iintv{1,n}$ by $\regheterdpsingleparam{\boldeps}{\boldsymbol{j}}$. Using Corollary \ref{corollary:mixture} and our insight resulting from the case $n=2$ in Theorem \ref{thm:main}, we conjecture that we can represent $\regcomposedp{\boldeps}{\bolddelta}{k}$ as a convex combination of heterogeneous composition regions $\regheterdpsingleparam{\boldeps}{\boldsymbol{j}}$ and $\regdp{0}{1}$.

%\begin{conjecture}%[Privacy region of the composition of $n$-DP constrained mechanisms]
%Let $\boldeps = (\eps_i)_{i=1}^n \in (\R^*_+)^n$ and $\bolddelta = (\delta_i)_{i=1}^n \in [0,1]^n$ such that all $(\eps_i,\delta_i)$-DP constraints are active. Then the privacy region of the composition of $k$ mechanisms which are $(\eps_i, \delta_i)$-DP for $i \in \Iintv{1,n}$ is\begin{align*}
%\regcomposedp{\boldeps}{\bolddelta}{k} = \Tilde{\delta} \regdp{0}{1} + (1-\Tilde{\delta}) \sum_{\substack{\boldsymbol{j} \in \Iintv{0,k}^n\\ \sum_{i=1}^n j_i = k}} \omega_{\boldsymbol{j}}\regheterdpsingleparam{\boldeps}{\boldsymbol{j}}
%\end{align*}
%for $\omega_{\boldsymbol{j}} = \binom{k}{j_1,\ldots,j_n}\prod_{i=1}^n\alpha_i^{j_i}$ for some functions $\alpha_i$, $\Tilde{\delta}$ of $\boldeps$, $\bolddelta$.
%\end{conjecture}

\newpage

%%%%%%%%%%%%%%%%%%%%%%%%%%%%%%%
% DO NOT TOUCH THE LINES BELOW IF WHAT YOU WANT TO DO IS HIDE THE APPENDIX. SEE LINE 82, JUST ABOVE \begin{document} INSTEAD.
\ifdefined\extended
    \appendix
    \subsection{Proof of Lemma \ref{lem:mixture}}\label{subsec:mixture_lemma_proof}
Let $I$ be the index of the randomly selected binary hypothesis test, i.e. a random variable supported on $\Iintv{1,n}$ such that $P_I(i) = \alpha_i$. Observe that, by the mixture definition of $\calH_m$, for $x \in \cbr{0,1}$,\[
\cbr{\calH_m = x} \cap \cbr{I = i} = \cbr{\calH^i = x}.
\]

First, we compute $\calH_m$'s type-I and type-II errors in terms of the $n$ hypothesis tests' errors. Given a decision rule $\hat{\calH}_m:\mc{Y} \times \Iintv{1,n} \to \cbr{0,1}$, let $\hat{\calH}^i(y) = \hat{\calH}_m(y,i)$, we compute using the law of total probability
\begin{align*}
\beta_\text{I}^m(\hat{\calH}_m) &= \bbP[\hat{\calH}_m(Y,I) = 1|\calH_m = 0]\\
&= \sum_{i=1}^n P_I(i)\bbP[\hat{\calH}_m(Y,i) = 1|\calH_m = 0, I=i] \\
&=\sum_{i=1}^n \alpha_i \bbP[\hat{\calH}^i(Y) = 1|\calH^i= 0]\\
&= \sum_{i=1}^n \alpha_i \beta_\mathrm{I}^i(\hat{\calH}^i).
\end{align*}
Similarly, by swapping the roles of the hypotheses in each test, it follows from the same computation that
\begin{equation*}
\beta_\text{II}^m(\hat\calH_m) = \sum_{i=1}^n \alpha_i  \beta_\mathrm{II}^i(\hat\calH^i).
\end{equation*}
Then we proceed to compute the trade-off function of $\calH_m$.
\begin{align}
f_m(t) &= \inf_{\substack{\hat{\calH}_m: \mathcal{Y}\times\Iintv{1,n} \to \cbr{0,1}}} \cbr{\beta_\text{II}^m(\hat{\calH}_m)\ | \  \beta_\text{I}^m(\hat{\calH}_m) \leq t} \nonumber \\
&= \inf_{\substack{\hat{\calH}^i: \mathcal{Y} \to \cbr{0,1} \\ i \in \Iintv{1,n}}} \cbr{\sum_{i=1}^n \alpha_i \beta_\text{II}^i(\hat{\calH}^i) \ | \ \sum_{i=1}^n \alpha_i \beta_\text{I}^i(\hat{\calH}^i) \leq t} \label{beta_IIs} \\
&= \inf_{\substack{\hat{\calH}^i: \mathcal{Y} \to \cbr{0,1} \\ i \in \Iintv{1,n}}}  \cbr{\sum_{i=1}^n \alpha_i f_i(\beta_\text{I}^i(\hat{\calH}^i)) \ | \ \sum_{i=1}^n \alpha_i \beta_\text{I}^i(\hat{\calH}^i) \leq t}. \label{tradeoffs}
\end{align}
The last step follows by observing that $\beta_\text{II}^i(\hat{\calH}^i) \geq f_i(\beta_\text{I}^i(\hat{\calH}^i))$ for $i \in \Iintv{1,n}$ and all $\hat{\calH}^i: \mathcal{Y} \to \cbr{0,1}$, hence (\ref{beta_IIs}) $\geq$ (\ref{tradeoffs}). To see (\ref{beta_IIs}) $\leq$ (\ref{tradeoffs}): by the definition of the infimum, $\forall \eps > 0 \ \exists \hat{\calH}^i: \mathcal{Y} \to \cbr{0,1}$ such that $f_i(\beta_\text{I}^i(\hat{\calH}^i)) \geq \beta_\text{II}^i(\hat{\calH}^i) - \eps$, thus
\begin{align*}
    \inf_{\substack{\hat{\calH}^i:\ \mathcal{Y} \to \cbr{0,1} \\ i \in \Iintv{1,n}}} &\cbr{\sum_{i=1}^n \alpha_i f_i(\beta_\text{I}^i(\hat{\calH}^i)) \ | \ \sum_{i=1}^n \alpha_i \beta_\text{I}^i(\hat{\calH}^i) \leq t}\\ &\geq \\
    \inf_{\substack{\hat{\calH}^i:\ \mathcal{Y} \to \cbr{0,1} \\ i \in \Iintv{1,n}}} &\cbr{\br{\sum_{i=1}^n \alpha_i \beta_\text{II}^i(\hat{\calH}^i)} -\eps\ | \ \sum_{i=1}^n \alpha_i \beta_\text{I}^i(\hat{\calH}^i) \leq t} \label{beta_IIs}.
\end{align*}
Letting $\eps$ go to 0 yields (\ref{tradeoffs}) $\geq$ (\ref{beta_IIs}).

Lastly, by noting that the set over which the mixture of trade-off functions is minimized above depends on $\hat{\calH}^i$ only through $\beta_\text{I}^i(\hat{\calH}^i)$, we set $\beta_\text{I}^i(\hat{\calH}^i) = t_i$ for $i \in \Iintv{1,n}$ and the trade-off function becomes
\begin{equation*}
    f_m(t) =  \inf_{\substack{t_i \in [0,1], \ i \in \Iintv{1,n}\\ \sum_{i=1}^n \alpha_i t_i \leq t}} \sum_{i=1}^n \alpha_i f_i(t_i).
\end{equation*}
The infimum is taken over all $t_i \in [0,1]$ as $\beta_\text{I}^i(\hat{\calH}^i)$ can take all values in $[0,1]$ due to the possible non-determinism of $\hat{\calH}^i$.

The trade-off functions $f_i$ are continuous as argued in \cite{fdp}, thus the function of $n$ variables $g_\alpha(t_1, \ldots, t_n) = \sum_{i=1}^n \alpha_i f_i(t_i)$ is continuous over the set being minimized upon, which is $S_\alpha = \cbr{(t_1,\ldots,t_n) \in [0,1]^n | \sum_{i=1}^n \alpha_it_i \leq t}.$ Observe that $S_\alpha$ is a closed set, because it is the pre-image of the closed set $[0,t]$ by the continuous function $k_\alpha(t_1,\ldots,t_n) = \sum_{i=1}^n \alpha_i t_i$. Also, $S_\alpha$ is bounded as $S_\alpha \subseteq [0,1]^n$. Hence the function $g_\alpha$ attains its minimum by the intermediate value theorem. In other words, the inf is a min.
\begin{equation*}
    f_m(t) =  \min_{\substack{t_i \in [0,1], \ i \in \Iintv{1,n}\\ \sum_{i=1}^n \alpha_i t_i \leq t}} \sum_{i=1}^n \alpha_i f_i(t_i).
\end{equation*}
Next, we can restrict the search space to the $(t_1,\ldots,t_n)$ such that $\sum_{i=1}^n \alpha_it_i = t$. Indeed, assume the min is achieved by $(t_1,\ldots,t_n)$ such that $\sum_{i=1}^n \alpha_it_i := t^* < t$, which can only happen if $t > 0$. Also assume $\alpha_i \neq 0$, as otherwise, the $i^\text{th}$ hypothesis test can be ignored. Then, because $f_i$ are non-increasing by Proposition 2 of \cite{fdp}, the min value is also achieved by $c_\lambda = \br{c_{\lambda_i} := t_i + \frac{\lambda_i}{\alpha_i}(t-t^*)}_{i=1}^n \in [0,1]^n$ with $\lambda = (\lambda_i)_{i=1}^n \in [0,1]^n$ chosen such that $c_\lambda \in [0,1]^n$ and $\sum_{i=1}^n \lambda_i=1$ - we show the existence of such a $\lambda$ later. This is because we are adding a non-negative quantity to each coordinate of the arguments of $f_i$, and as such, since they are non-increasing, they will remain at their minimal values. Thus, $\sum_{i=1}^n \alpha_i f_i(c_{\lambda_i}) = \sum_{i=1}^n \alpha_i f_i(t_i) = \min_{\substack{t_i' \in [0,1], \ i \in \Iintv{1,n}\\ \sum_{i=1}^n \alpha_i t_i' \leq t}} \sum_{i=1}^n \alpha_i f_i(t_i').$ Also,
\begin{align*}
    \sum_{i=1}^n \alpha_i c_{\lambda_i} &= \sum_{i=1}^n\alpha_i t_i + \alpha_i\frac{\lambda_i}{\alpha_i}(t-t^*)\\ &= \br{\sum_{i=1}^n \alpha_i t_i} + \sum_{i=1}^n \lambda_i (t-t^*)\\
    &= t^* + t-t^*\\
    &= t.
\end{align*}
Hence the constraint $\sum_{i=1}^n \alpha_i t_i' \leq t$ can become the more specific constraint $\sum_{i=1}^n \alpha_i t_i' = t$ in the minimization problem defining $f_m(t)$.

A $\lambda \in [0,1]^n$ such that $c_\lambda \in [0,1]^n$ exists. To show this, we simplify the constraints on $\lambda$.
\begin{align*}
&\exists \lambda \begin{cases}
    t_i + \frac{\lambda_i}{\alpha_i}(t-t^*) \in [0,1]  &i \in \Iintv{1,n}\\
    \lambda_i \in [0,1] &i \in \Iintv{1,n}\\
    \sum_{i=1}^n \lambda_i = 1
\end{cases}\\ \iff &\exists \lambda \begin{cases}
    \lambda_i \in \left[-\alpha_i\frac{t_i}{t-t^*},\alpha_i\frac{1-t_i}{t-t^*}\right]  &i \in \Iintv{1,n}\\
    \lambda_i \in [0,1] &i \in \Iintv{1,n}\\
    \sum_{i=1}^n \lambda_i = 1
\end{cases} \\
\iff &\exists \lambda \begin{cases}
    \lambda_i \in \left[0,\alpha_i\frac{1-t_i}{t-t^*}\right]  &i \in \Iintv{1,n}\\
    \lambda_i \in [0,1] &i \in \Iintv{1,n}\\
    \sum_{i=1}^n \lambda_i = 1
\end{cases} \\
\iff &\exists \lambda \begin{cases}
    \lambda_i \in \left[0,\min\br{1,\alpha_i\frac{1-t_i}{t-t^*}}\right]  &i \in \Iintv{1,n}\\
    \sum_{i=1}^n \lambda_i = 1.
\end{cases} 
\end{align*}
Let $\mc{S} = \prod_{i=1}^n \left[0,\min\br{1,\alpha_i\frac{1-t_i}{t-t^*}}\right]$. Define $\sigma: \mc{S} \to \bbR$ such that $\sigma(\lambda) = \sum_{i=1}^n \lambda_i$. Also let $\lambda^* = \br{\min\br{1,\alpha_i\frac{1-t_i}{t-t^*}}}_{i=1}^n \in \mc{S}.$
\begin{itemize}
    \item $\sigma(0^n) = 0.$
    \item If $\exists j \in \Iintv{1,n}$ such that $\min\br{1,\alpha_j\frac{1-t_j}{t-t^*}} = 1$, then $\sigma(\lambda^*) = \sum_{i=1}^n \min\br{1,\alpha_i\frac{1-t_i}{t-t^*}} := \sigma_1 \geq \min\br{1,\alpha_j\frac{1-t_j}{t-t^*}} = 1$.
    \item Else, 
    \begin{align*}
        \sigma(\lambda^*) &= \sum_{i=1}^n \alpha_i\frac{1-t_i}{t-t^*} \\
        &=  \frac{\sum_{i=1}^n \alpha_i-\alpha_i t_i}{t-t^*}\\ &= \frac{1 - t^*}{t-t^*} := \sigma_2 \geq 1
    \end{align*}
    because $1\geq t \implies 1-t^* \geq t-t^*$.
\end{itemize}
Then, $\sigma$ is a continuous function of closed and bounded domain $\mc
S = \prod_{i=1}^n \left[0,\min\br{1,\alpha_i\frac{1-t_i}{t-t^*}}\right]$. By the intermediate value theorem, since $\sigma(0) = 0$ and $\sigma(\lambda^*) \geq \min(\sigma_1,\sigma_2) \geq 1$, there exists $\lambda \in \mc{S}$ such that $\sigma(\lambda) = 1$ since $1 \in [\sigma(0), \sigma(\lambda^*)]$. 

\subsection{Proof of Corollary \ref{corollary:mixture}}\label{subsec:corollary_mixture_proof}
We proceed by double inclusion.
\subsubsection{$\sum_{i=1}^n\alpha_i \regtradeoffindex{f}{i}  \subseteq \regtradeoffindex{f}{m}$}
Let $\beta^i = \begin{bmatrix}
    \beta_\mathrm{I}^i \\ \beta_\mathrm{II}^i
\end{bmatrix} \in \regtradeoffindex{f}{i}$, $i \in \Iintv{1,n}$. Letting $\beta_\mathrm{I}^m = \sum_{i=1}^n \alpha_i \beta_\mathrm{I}^i$, first observe that
\begin{equation*}
1-\beta_\mathrm{I}^m = 1-\br{\sum_{i=1}^n \alpha_i \beta_\mathrm{I}^i}\geq \beta_\mathrm{II}^m := \sum_{i=1}^n \alpha_i \beta_\mathrm{II}^i \geq \sum_{i=1}^n \alpha_i f_i(\beta_\mathrm{I}^i).
\end{equation*}
This is due to the definition of $\regtradeoffindex{f}{i}$, $1-\beta_\mathrm{I}^i \geq \beta_\mathrm{II}^i \geq f_i(\beta_\mathrm{I}^i)$ for $i \in \Iintv{1,n}$. 

By the previous lemma we have that
\begin{equation*}
    \beta_\mathrm{II}^m \geq \sum_{i=1}^n \alpha_i f_i(\beta_\mathrm{I}^i)\geq f_m(\beta_\mathrm{I}^m).
\end{equation*}
Indeed, $(\beta_\mathrm{I}^i)_{i=1}^n$ is a feasible solution of the minimization defining $f_m(\beta_\mathrm{I}^m)$.\\
Writing $\beta^m = \begin{bmatrix}
    \beta_\mathrm{I}^m \\ \beta_\mathrm{II}^m
\end{bmatrix}$, by construction $\sum_{i=1}^n\alpha_i \beta^i = \beta^m$. $f_m(\beta_\mathrm{I}^m) \leq \beta_\mathrm{II}^m \leq 1-\beta_\mathrm{I}^m$ as shown above. Thus, $\beta^m = \sum_{i=1}^n\alpha_i \beta^i \in \regtradeoffindex{f}{m}.$

\subsubsection{$\regtradeoffindex{f}{m} \subseteq \sum_{i=1}^n\alpha_i \regtradeoffindex{f}{i}$}
Let $\beta^m = \begin{bmatrix}
    \beta_\mathrm{I}^m \\ \beta_\mathrm{II}^m
\end{bmatrix} \in \regtradeoffindex{f}{m}$. \begin{enumerate}[i.]
    \item If $\beta^m$ is on the lower boundary of $\regtradeoffindex{f}{m}$, i.e. $\beta_\mathrm{II}^m = f_m(\beta_\mathrm{I}^m)$, then by the previous lemma, since $f_m(\beta_\mathrm{I}^m) = \min \cbr{\sum_{i=1}^n \alpha_i f_i(\beta_\mathrm{I}^i) \ | \ \sum_{i=1}^n \alpha_i \beta_\mathrm{I}^i = \beta_\mathrm{I}^m, \beta_\mathrm{I}^i \in [0,1] }$, there exists $(\beta_\mathrm{I}^i)_{i=1}^n \in [0,1]^n$ such that $\sum_{i=1}^n \alpha_i \beta_\mathrm{I}^i = \beta_\mathrm{I}^m$ and $\sum_{i=1}^n \alpha_i f_i(\beta_\mathrm{I}^i) = \beta_\mathrm{I}^m = f_m(\beta_\mathrm{I}^m)$. Thus \begin{equation*}
        \begin{bmatrix}
    \beta_\mathrm{I}^m \\ f_m(\beta_\mathrm{I}^m)
\end{bmatrix} = \sum_{i=1}^n \alpha_i \underbrace{\begin{bmatrix}
    \beta_\mathrm{I}^i \\ f_i(\beta_\mathrm{I}^i)
\end{bmatrix}}_{\in \regtradeoffindex{f}{i}} \in \sum_{i=1}^n \alpha_i \regtradeoffindex{f}{i}.
    \end{equation*}
    \item If $\beta^m$ is on or above the highest of the lower boundaries of $\cbr{\regtradeoffindex{f}{i}}_{i=1}^n$, i.e. $\beta_\mathrm{II}^m \geq \max\br{f_i(\beta_\mathrm{I}^m)}_{i=1}^n$, then $\beta^m$ is in $\cap_{i=1}^n \regtradeoffindex{f}{i}$. Thus $\beta^m = \sum_{i=1}^n \alpha_i \beta^m \in \sum_{i=1}^n \alpha_i \regtradeoffindex{f}{i}.$
    \item If $\beta^m$ sits between the lower boundary of $\regtradeoffindex{f}{m}$ and the highest of the lower boundaries of $\cbr{\regtradeoffindex{f}{i}}_{i=1}^n$, i.e. $f_m(\beta_\mathrm{I}^m) \leq \beta_\mathrm{II}^m \leq \max\br{f_i(\beta_\mathrm{I}^m)}_{i=1}^n = f_{i^*}(\beta_\mathrm{I}^m)$ with $i^* \in \arg\max_{i\in\Iintv{1,n}} f_i(\beta_\mathrm{I}^m)$, observe that in this case $\beta^m$ sits in the vertical segment \begin{equation*}
        S = \cbr{\lambda\begin{bmatrix}
    \beta_\mathrm{I}^m \\ f_m(\beta_\mathrm{I}^m)
\end{bmatrix} + (1-\lambda)\begin{bmatrix}
    \beta_\mathrm{I}^m \\ f_{i^*}(\beta_\mathrm{I}^m)
\end{bmatrix} \ | \ \lambda \in [0,1]}.
    \end{equation*}
    In the two previous points we have shown that the two ends of the segment $S$ are elements of $\sum_{i=1}^n \alpha_i \regtradeoffindex{f}{i}$. $\sum_{i=1}^n \alpha_i \regtradeoffindex{f}{i}$ is a convex combination of convex sets $\regtradeoffindex{f}{i}$, thus $\sum_{i=1}^n \alpha_i \regtradeoffindex{f}{i}$ is convex itself \cite{boyd_convex_optimization}. Thus $S \subseteq \sum_{i=1}^n \alpha_i \regtradeoffindex{f}{i}$, in particular $\beta^m \in \sum_{i=1}^n \alpha_i \regtradeoffindex{f}{i}$.
\end{enumerate}
We have shown by case disjunction that $\beta^m \in \sum_{i=1}^n \alpha_i \regtradeoffindex{f}{i}$.

\subsection{Proof of Theorem \ref{thm:heter}}\label{subsec:proof_heter}
Define distributions $P_i^1$ on $\cbr{0,3}$ and $P_i^2$ on $\cbr{1,2}$: 
\begin{equation*}
     P_i^1(x) = \begin{cases}
        \frac{e^{\eps_1}}{e^{\eps_1} + 1} &\text{ if } (i = 0, x=0) \text{ or }(i=1, x= 3),\\
        \frac{1}{e^{\eps_1} + 1} &\text{ if } (i = 0, x=3) \text{ or }(i=1, x= 0),
    \end{cases}
\end{equation*}
and 
\begin{equation*}
    P_i^2(x) = \begin{cases}
        \frac{e^{\eps_2}}{e^{\eps_2} + 1} &\text{ if } (i = 0, x=1) \text{ or }(i=1, x= 2),\\
        \frac{1}{e^{\eps_2} + 1} &\text{ if } (i = 0, x=2) \text{ or }(i=1, x= 1).
    \end{cases}
\end{equation*}
For $j \in \cbr{1,2}$, let the mechanism $M^j$ output $X_0^j \sim P_0^j$ in the case of the null hypothesis, $X_1^j \sim P_1^j$ else. Thus $M^j$ is $\eps_j$-DP binary randomized response. Let $M^{x,y}$ be the composition of $x$ replicas of $M^1$ concatenated with $y$ replicas of $M^2$. Then $M^{x,y}$ outputs $X_0 = ((X_0^1)^x,(X_0^2)^y) \sim \Tilde{P}_0 = (P_0^1)^x(P_0^2)^y$ in the case of the null hypothesis, and $X_1 = ((X_1^1)^x,(X_1^2)^y) \sim \Tilde{P}_1 = (P_1^1)^x(P_1^2)^y$ else, where $\Tilde{P}_0$ and $\Tilde{P}_1$ are distributions on $\Iintv{0,3}^{x+y}$.\\
It remains to compute the privacy region of $M$ and to show that it is the largest composition region in this case. 
\subsubsection{Privacy region of $M$}\label{subsubsec:privacy_reg_heter}
Let $k = x+y$. Let $s^k \in \cbr{0,1,2,3}^k$, and define
\begin{itemize}
    \item $a := a(s^k) = \abs{\cbr{i \  | \ s_i = 0}}$,
    \item $b := b(s^k) = \abs{\cbr{i \  | \ s_i = 1}}$,
    \item $c := c(s^k) = \abs{\cbr{i \  | \ s_i = 2}}$,
    \item $d := d(s^k) = \abs{\cbr{i \  | \ s_i = 3}}$.
\end{itemize}
Thus, we have that $(a,b,c,d) \in \Iintv{0,k}^4$ and $a+b+c+d = k$. Additionally, $x = a+d$, $y=b+c$.\\

For $s^k \in \cbr{0,1,2,3}^k$ such that $(s_1,\ldots,s_x) \in \cbr{0,3}^x$ and $(s_{x+1}, \ldots, s_{x+y}) \in \cbr{1,2}^y$, we compute the likelihoods $\Tilde{P}_j(s^k)$:\begin{align*}
\Tilde{P}_0(s^k) &= \br{\prod_{i=1}^x P_0^1(s_i)} \br{\prod_{i=1}^{y} P_0^2(s_{x+i})}\\
&=  \br{\frac{1}{e^{\eps_1}+1}}^{x}e^{a\eps_1}\br{\frac{1}{e^{\eps_2}+1}}^{y}e^{b\eps_2}\\
\Tilde{P}_0(s^k)&= \br{\frac{1}{e^{\eps_1}+1}}^{x}\br{\frac{1}{e^{\eps_2}+1}}^{y}e^{a\eps_1+b\eps_2},
\end{align*}
and similarly\[
\Tilde{P}_1(s^k) = \br{\frac{1}{e^{\eps_1}+1}}^{x}\br{\frac{1}{e^{\eps_2}+1}}^{y}e^{d\eps_1+c\eps_2}.
\]
Drawing inspiration from the proofs from \cite{compositiontheorem,dptv}, we are looking for all pairs $(\eps,\delta)$ such that, for the hypothesis testing problem $\calH(\Tilde{P}_0, \Tilde{P}_1)$,
\begin{align}
1-\Tilde{P}_0(A) = \beta_\mathrm{I} \geq -e^\eps \beta_\mathrm{II} + 1-\delta = -e^\eps \Tilde{P}_1(A) + 1-\delta\label{equ:hyptest_dp_base}
\end{align}
for any non-rejection set $A \subseteq \cbr{0,3}^x \times \cbr{1,2}^y$. For each $(\eps,\delta)$ found, the mechanism $M$ is $(\eps,\delta)$-DP as Equation \ref{equ:hyptest_dp_base} is the binary hypothesis testing characterization of $(\eps,\delta)$-DP \cite{compositiontheorem}. Since the privacy region of $M$ is convex, it is entirely defined by all such line-offset pairs $(\eps,\delta)$ \cite{compositiontheorem}.\\
Firstly, we determine all the possible slopes $e^\eps$. By the Neyman-Pearson lemma, the optimal tests for the hypothesis testing problem $\calH(\Tilde{P}_0, \Tilde{P}_1)$, minimizing $\beta_\mathrm{II}$ for a fixed maximum value of $\beta_\mathrm{I}$, are of the form $A(\eps)  = \cbr{s^k\ | \ \frac{\Tilde{P}_0(s^k)}{\Tilde{P}_1(s^k)} \geq e^\eps} = \cbr{s^k\ | \ {\Tilde{P}_0(s^k)}-e^\eps{\Tilde{P}_1(s^k)} \geq 0}$. We are thus interested in\[
    \frac{\tilde{P}_0(s^k)}{\tilde{P}_1(s^k)} = e^{\eps_1(a-d)+\eps_2(b-c)},
\]
looking for all the values $e^\eps$ this likelihood ratio can take. With $x=a+d$ and $y=b+c$, we can write
\begin{equation*}
    \frac{\tilde{P}_0(s^k)}{\tilde{P}_1(s^k)} = e^{\eps_1(a+d-2d)+\eps_2(b+c-2c)} = e^{\eps_1(x-2a^*)+\eps_2(y-2b^*)}
\end{equation*}
where $d:=a^* \in \Iintv{0,x}, c:=b^* \in \Iintv{0,y}$.\\
By the symmetry of DP, it is sufficient to keep the slopes $\eps_{a^*,b^*}^{x,y}$ with $\eps_{a^*,b^*}^{x,y} = \eps_1(x-2a^*)+\eps_2(y-2b^*) \geq 0$.\\
Secondly, we find the matching $\delta$ for each $\eps$. Rewrite \eqref{equ:hyptest_dp_base} as \[
1-\delta \leq \beta_\mathrm{I} + e^{\eps}\beta_\mathrm{II} = 1-\tilde{P}_0(A)+e^\eps\tilde{P}_1(A).
\]
For an optimal non-rejection set $A(\eps)$, the right-hand side is minimized and the inequality becomes an equality:\begin{align*}
1-\delta &= 1-\tilde{P}_0(A(\eps))+e^\eps\tilde{P}_1(A(\eps))\\
&= \min_{A \subseteq \cbr{0,3}^x \times \cbr{1,2}^y} 1-\tilde{P}_0(A)+e^\eps\tilde{P}_1(A),
\end{align*}
or equivalently,\[
\delta = \max_{A \subseteq \cbr{0,3}^x \times \cbr{1,2}^y} \tilde{P}_0(A)-e^\eps\tilde{P}_1(A).
\]
Consequently, for each slope $\eps_{a^*,b^*}^{x,y} = \eps_1(x-2a^*)+\eps_2(y-2b^*)$ such that $(a^*,b^*) \in S(\eps_1,\eps_2;x,y) = \cbr{(a^*,b^*) \ | \ \eps_1(x-2a^*)+\eps_2(y-2b^*) \geq 0} \in \Iintv{0,x} \times \Iintv{0,y}$, we can compute the corresponding offset $\delta_{a^*,b^*}^{x,y}$.
\begin{align*}
    \delta_{a^*,b^*}^{x,y} &= \max_{A \subseteq \cbr{0,3}^x \times \cbr{1,2}^y} \Tilde{P}_0(A) - e^{\eps_{a^*,b^*}^{x,y}} \Tilde{P}_1(A)\\
    &= \sum_{s^k \in \cbr{0,3}^x \times \cbr{1,2}^y}\max\br{0, \Tilde{P}_0(s^k) - e^{\eps_{a^*,b^*}^{x,y}}\Tilde{P}_1(s^k)}
\end{align*}
Observe that $\Tilde{P}_j$ depend on $s^k$ only through $a(s^k)$ and $b(s^k)$. Thus, if $a(s^k) = a$ and $b(s^k) = b$,\[
\Tilde{P}_j(s^k) = \Tilde{P}_j(0^a3^{x-a}1^b2^{y-b}) := \Tilde{P}_j(s^{a,b}).
\]
For a given pair $(a,b) \in \Iintv{0,x} \times \Iintv{0,y}$, there exists $\binom{x}{a}\binom{y}{b}$ sequences $s^k \in \cbr{0,3}^x \times \cbr{1,2}^y$ such that $a(s^k)=a$ and $b(s^k) = b$. Then, by grouping all $s^k$ with same values $a(s^k)$ and $b(s^k)$,
\[
     \delta_{a^*,b^*}^{x,y} =\sum_{b=0}^y\sum_{a=0}^x \binom{x}{a}\binom{y}{b}\max\br{0, \Tilde{P}_0(s^{a,b}) - e^{\eps_{a^*,b^*}^{x,y}}\Tilde{P}_1(s^{a,b})}.
\]
To simplify the max, observe that
\begin{align*}
    &\Tilde{P}_0(0^a3^{x-a}1^b2^{y-b}) - e^{\eps_{a^*,b^*}^{x,y}}\Tilde{P}_1(0^a3^{x-a}1^b2^{y-b}) \geq 0 \\
    \iff &e^{a\eps_1 + b\eps_2}-e^{\eps_{a^*,b^*}^{x,y}} e^{\eps_1(x-a) + \eps_2(y-b)} \geq 0\\
    \iff &e^{a\eps_1 + b\eps_2}- e^{\eps_1(x-2a^* + x-a) + \eps_2(y-2b^*+y-b)} \geq 0\\
    \iff &e^{a\eps_1 + b\eps_2}- e^{\eps_1(2(x-a^*)-a) + \eps_2(2(y-b^*)-b)} \geq 0\\
    \iff &a\eps_1 + b\eps_2\geq \eps_1(2(x-a^*)-a) + \eps_2(2(y-b^*)-b) \\
    \iff &2a\eps_1 + 2b\eps_2 \geq 2(x-a^*)\eps_1 + 2(y-b^*)\eps_2\\
    \iff &a\eps_1 + b\eps_2 \geq (x-a^*)\eps_1 + (y-b^*)\eps_2\\
    \iff &a \geq (y-b^*-b)\frac{\eps_2}{\eps_1} + (x-a^*).
\end{align*}
Consequently,
\begin{align}
  \delta_{a^*,b^*}^{x,y} &= \sum_{b=0}^y\sum_{a=a_0(b)}^x \binom{x}{a}\binom{y}{b} \Tilde{P}_0(s^{a,b}) - e^{\eps_{a^*,b^*}^{x,y}}\Tilde{P}_1(s^{a,b})\nonumber\\
  &=  \br{\frac{1}{e^{\eps_1}+1}}^x\br{\frac{1}{e^{\eps_2}+1}}^y \sum_{b=0}^y\sum_{a=a_0(b)}^x \binom{x}{a}\binom{y}{b}\cdot \nonumber\\ &\br{e^{a\eps_1 + b\eps_2}- e^{\eps_1(2(x-a^*)-a) + \eps_2(2(y-b^*)-b)}} \label{delta_a*b*}
\end{align}
where
\begin{equation*}
    a_0(b) = \max\br{0,\lceil (y-b^*-b)\frac{\eps_2}{\eps_1} + (x-a^*) \rceil}.
\end{equation*}
This was the last step to prove the correctness of Algorithm \ref{alg:het_comp}.

\subsubsection{Largest region}\label{subsubsec:converse_proof}
We again proceed similarly as to \cite{compositiontheorem,dptv}. Each of the first $x$ mechanisms given above, $\eps_1$-binary randomized response, achieves the full $\eps_1$-DP region, and the same goes for the last $y$ mechanisms that achieve the full $\eps_2$-DP region \cite{compositiontheorem}. Thus, by Theorem 10 of \cite{blackwell}, any set of $x$ $\eps_1$-DP mechanisms and $y$ $\eps_2$-DP mechanisms can be simulated through $x$ instances of $\eps_1$-binary randomized response and $\eps_2$-binary randomized response, and the same converse proof as \cite{compositiontheorem} can be followed.

\subsection{Proof of Theorem \ref{thm:main}, case $\delta_1 > 0$}\label{subsubsec:extension_delta_1}
Define
\begin{equation*}
    P_0(x) = \begin{cases}
        \delta_1 &\quad x=-1, \\
        (1-\delta_1)(1-\alpha)\frac{e^{\eps_1}}{e^{\eps_1} + 1} &\quad x=0,\\
        (1-\delta_1) \alpha \frac{e^{\eps_2}}{e^{\eps_2} + 1} &\quad x=1,\\
        (1-\delta_1) \alpha \frac{1}{e^{\eps_2} + 1} &\quad x=2,\\
        (1-\delta_1)(1-\alpha)\frac{1}{e^{\eps_1} + 1} &\quad x=3, \\
        0 &\quad x=4,
    \end{cases}
\end{equation*}
and
\begin{equation*}
    P_1(x) = \begin{cases}
        0 &\quad x=-1, \\
        (1-\delta_1) (1-\alpha)\frac{1}{e^{\eps_1} + 1} &\quad x=0, \\
        (1-\delta_1) \alpha \frac{1}{e^{\eps_2} + 1} &\quad x=1,\\
        (1-\delta_1) \alpha \frac{e^{\eps_2}}{e^{\eps_2} + 1} &\quad x=2,\\
        (1-\delta_1) (1-\alpha)\frac{e^{\eps_1}}{e^{\eps_1} + 1} &\quad x=3 \\
        \delta_1 &\quad x=4.
    \end{cases}
\end{equation*}
Thus, with the same notation as the case $\delta_1 = 0$, \[
    P_i(x) = \begin{cases}
            (1-\delta_1)(1-\alpha)P_i^1(x) &\text{ if } x \in\cbr{0,3},\\
        (1-\delta_1)\alpha P_i^2(x) &\text{ if } x \in\cbr{1,2},\\
        \delta_1P_i^3(x) &\text{ if } x \in\cbr{-1,4},
    \end{cases}
\]
where\[
P_0^3(x) = \mathbf{1}(x=-1), \ P_1^3(x) = \mathbf{1}(x=4).
\]
Hence, in this case, $M$ is the mixture of 3 mechanisms, picking $M^1$ with probability $(1-\delta_1)(1-\alpha)$, $M^2$ with probability $(1-\delta_1)\alpha$ and $M^3$ with probability $\delta_1$. If $M_k$, which is $M$ composed with itself $k$ times, picks $M^3$ once or more, then the hypothesis test has probability 0 of both type I and type II errors as the true hypothesis is revealed. The probability that $M_k$ never picks $M^3$ in $k$ tries is $1-(1-\delta_1)^k$, and the remaining probability weight $(1-\delta_1)^k$ distributes across a mixture of the same mechanisms as in the $\delta_1 = 0$ case.\\

The above computes the privacy region for the $k$-composition of $M$. It remains to show that this is the largest possible region. First, we note that $f(P_0, P_1) = \max\{f_{\eps_1, \delta_1}, f_{\eps_2, \delta_2}\}$. This can be seen by the Neyman-Pearson lemma and calculating the achievable error pairs for the decision rules 
\begin{equation*}
    \hatcalH_\tau(x) = \begin{cases}
        0 &\quad P_0(x)\geq \tau P_1(x), \\
        1 &\quad P_0(x) < \tau P_1(x),
    \end{cases}
\end{equation*}
for $\tau \geq 0$. The rest of the proof  follows similarly to section \ref{subsubsec:converse_proof} just above.

\subsection{Proof of Theorem \ref{thm:alt-main}}\label{subsec:proof_alt_main}
\subsubsection{Likelihood ratio introduction}
We first cover the case $\delta_1 = 0$, then one can extend to the case $\delta_1 > 0$ using the same machinery as in section \ref{subsubsec:extension_delta_1}. Consider again the same mechanism as introduced in the proof of Theorem \ref{thm:main}, written in full detail this time: $M$ outputs $X_i \sim P_i$ under hypothesis $\calH_i$, where, for $i \in \cbr{0,1}$,
\begin{equation*}
    P_0(x) = \begin{cases}
        (1-\alpha)\frac{e^{\eps_1}}{e^{\eps_1} + 1} &\quad x=0,\\
        \alpha \frac{e^{\eps_2}}{e^{\eps_2} + 1} &\quad x=1,\\
        \alpha \frac{1}{e^{\eps_2} + 1} &\quad x=2,\\
        (1-\alpha)\frac{1}{e^{\eps_1} + 1} &\quad x=3,
    \end{cases}
\end{equation*}
and
\begin{equation*}
    P_1(x) = \begin{cases}
        (1-\alpha)\frac{1}{e^{\eps_1} + 1} &\quad x=0, \\
        \alpha \frac{1}{e^{\eps_2} + 1} &\quad x=1,\\
        \alpha \frac{e^{\eps_2}}{e^{\eps_2} + 1} &\quad x=2,\\
        (1-\alpha)\frac{e^{\eps_1}}{e^{\eps_1} + 1} &\quad x=3.
    \end{cases}
\end{equation*}
Consider the $k$-composition of $M$ with itself, $M_k$. $M_k$ outputs $(X_i)^k \sim \Tilde{P}_i := P_i^k$ under the hypothesis $\calH_i$, $i \in \cbr{0,1}$. With the same notation as in \ref{subsubsec:privacy_reg_heter}, for $s^k \in \Iintv{0,3}^k$,
\begin{align*}
    \tilde{P}_0(s^k) &= \prod_{i=1}^k P_0(s_i) = P_0(0)^aP_0(1)^bP_0(2)^cP_0(3)^d\\ &=  \br{\frac{1-\alpha}{e^{\eps_1}+1}}^{a+d}\br{\frac{\alpha}{e^{\eps_2}+1}}^{b+c}e^{a\eps_1 + b\eps_2},
\end{align*}
and 
\begin{align*}
    \tilde{P}_1(s^k) &= \prod_{i=1}^k P_1(s_i) = P_1(0)^aP_1(1)^bP_1(2)^cP_1(3)^d\\ &=  \br{\frac{1-\alpha}{e^{\eps_1}+1}}^{a+d}\br{\frac{\alpha}{e^{\eps_2}+1}}^{b+c}e^{d\eps_1 + c\eps_2}.
\end{align*}
As in the computation of the privacy region in the proof \ref{subsubsec:privacy_reg_heter}, we are interested in the slope-offset achieved by the likelihood ratio $\frac{\Tilde{P}_0}{\Tilde{P}_1}$, where\begin{equation*}
    \frac{\tilde{P}_0(s^k)}{\tilde{P}_1(s^k)} = e^{\eps_1(a-d)+\eps_2(b-c)}.
\end{equation*}
\subsubsection{Set of slopes}
We seek to find the image of the function $f: A \to \R$, $f(a,b,c,d) =\eps_1(a-d)+\eps_2(b-c)$, for fixed values of $\eps_1, \eps_2$. Here, 
\begin{equation*}
    A = \cbr{(a,b,c,d) \in \Iintv{0,k}^4 \ | \ a+b+c+d = k}.
\end{equation*}
Note that it suffices to find the values $(a,b,c,d)$ for which $f(a,b,c,d) \geq 0$. Indeed, if $f(a,b,c,d) = \eps_1(a-d)+\eps_2(b-c) < 0$, then $f(d,c,b,a) = \eps_1(d-a)+\eps_2(c-b) = -\sbr{\eps_1(a-d)+\eps_2(b-c)} > 0$.

We are interested in the values that $f$ attains, i.e. in $\imag{f}$. First, let
\begin{align*}
    \Tilde{f}: \ \Tilde{A} &\to \R\\
    (a-d,b-c) &\mapsto \eps_1(a-d) + \eps_2(b-c),
\end{align*}
where
\begin{align*}
\Tilde{A} &= \cbr{(a-d, b-c) \ | \ \substack{(a,b,c,d) \in \Iintv{0,k}^4\\ \text{ and } a+b+c+d=k}}\\
&= \cbr{(a-d, b-c) \ | \ (a,b,c,d) \in A}.
\end{align*}
It is easy to see that $f(a,b,c,d) = \Tilde{f}(a-d, b-c)$ and $\imag{f} = \imag{\Tilde{f}}$.

We rewrite the set $\Tilde{A}$ as follows. Define
\begin{equation*}
    B = \cbr{(x,y) \in \mathbb{Z}^2 \ | \ |x|+|y| \leq k, \ x + y \cong k \text{ mod } 2}.
\end{equation*}
then let $g: B \to \R$ such that \begin{equation*}
    g(x,y) = \eps_1x + \eps_2y.
\end{equation*}
It so happens that $\Tilde{A} = B$ as we prove in the sequel, thus $\imag{f} = \imag{\Tilde{f}} = \imag{g}$. 

To prove $\Tilde{A} = B$ proceed by double inclusion.
\begin{enumerate}
    \item To prove $\Tilde{A} \subseteq B$: let $(a-d,b-c) \in \Tilde{A}$. Define $x = a-d$ and $y = b-c$. Then by triangular inequality
    \[
    |x|+|y| = |a-d| + |b-c| \leq |a| + |d|+|b|+|c| = k.
    \]
    Also, observe that $\forall z \in \mathbb{Z}$, $z \cong -z \text{ mod }2$. Thus
    \[
    x+y = a-d+b-c\cong a+d+b+c = k \text{ mod } 2.
    \]
    We conclude that $(a-d,b-c) = (x,y) \in B$.
    \item To prove $B \subseteq \Tilde{A}:$ let $(x,y) \in B$. 
    \begin{itemize}
        \item \underline{If $x \geq 0$ and $y \geq 0$}: $|x|+|y| = x+y \leq k$. Let $s = k-(x+y) \in \Iintv{0,k}$. Observe that $s \cong 0 \text{ mod }2.$ Indeed,
        \begin{align*}
        \begin{cases}
            x+y\cong k \text{ mod }2\\
            x+y+s=k
        \end{cases} &\implies \begin{cases}
        x+y\cong k \text{ mod }2\\
            x+y+s \cong k \text{ mod }2
        \end{cases}\\ &\implies  s\cong 0 \text{ mod }2.
        \end{align*}
        Let $s = 2q$, $q \in \Iintv{0,k}$. Then we can define $a = x+q$, $d = q$, $b = y$, $c = 0$, and we have $(x,y) = (a-d, b-c) \in A$. In particular, observe that $a \in \Iintv{0,k}$ since $0 \leq x+q = a \leq x+q+y+q = k$.
        \item \underline{If $x < 0$ and $y \geq 0$}: similarly, $|x|+|y| = y-x \leq k$. Define $s = k - y+x \geq 0$. By the same argument as before, $s \cong 0 \text{ mod }2$, as $-x \cong x \text{ mod }2$. Thus we can write $s = 2q$, $q \in \Iintv{0,k}$. Letting $d = |x|$, $a=0$, $b = y+q$ and $c = q$, we obtain the desired result. Here as well, observe in particular that $b \in \Iintv{0,k}$ as $0 \leq y+q = b$, then assume for the sake of contradiction that $y+q > k$. Then $y+q+q > k +q \iff y+s > k+q \iff k+x > k+q \iff x > q$, but $x < 0$ and $q \geq 0$, leading to a contradiction.
        \item The other cases can easily be deducted from these first two cases. For instance, when $x \leq 0$ and $y \leq 0$, we simply have to swap the roles of $a,d$ and $b,c$ in the first point.
    \end{itemize}
\end{enumerate}
Thus, we find 
\begin{align*}
    \text{Slopes} &= \imag{f} = \imag{g}\\ 
    &= \cbr{\eps_1 x + \eps_2 y \ | \ \substack{(x,y) \in \mathbb{Z}^2, \ |x|+|y| \leq k,\\ x+y\cong k \text{ mod } 2}}.
\end{align*}
From this, the next step is to see
\begin{equation*}
\text{Slopes} = \cbr{\eps_1(u+v-k) + \eps_2(u-v) \ | \ (u,v)\in \Iintv{0,k}^2}.
\end{equation*}
To prove this, we will show that the following function is a bijection:
\begin{align*}
    h: B &\to \Iintv{0,k}^2\\
      (x,y) &\mapsto (u,v) = \br{\frac{k+x+y}{2}, \frac{k+x-y}{2}}.
\end{align*}
It suffices to check that $h$ is well defined and that its inverse $h^{-1}(u,v) = (u+v-k, u-v)$ is well defined as well. To see this,
\begin{align*}
|x| +|y| \leq k &\implies  \begin{cases}
    -k \leq x+y \leq k\\
    -k \leq x-y \leq k
\end{cases}\\
&\implies \begin{cases}
    0\leq k+x+y \leq 2k\\
    0 \leq k+x-y \leq 2k
\end{cases}\\
&\implies \begin{cases}
    0\leq \frac{k+x+y}{2} \leq k\\
    0 \leq \frac{k+x-y}{2} \leq k.
\end{cases}
\end{align*}
These two fractions are integers because $x+y\cong x-y \cong k \text{ mod }2$. We can similarly show that $h^{-1}$ is well defined as follows. First, note that $(u+v-k)+(u-v) = 2u-k\cong -k \cong k \text{ mod }2$. Also,
\begin{itemize}
    \item If $u+v \geq k$, then \[
    |u+v-k| + |u-v| = \begin{cases}
        u+v-k+u-v\\ = 2u-k \leq k \text{ if } u \geq v,\\
        u+v-k+v-u \\= 2v-k \leq k \text{ if } v\geq u.
    \end{cases}
    \]
    \item If $u+v < k$, then \[
    |u+v-k| + |u-v| = \begin{cases}
        k-u-v+u-v\\ = k-2v \leq k \text{ if } u \geq v,\\
        k-u-v+v-u\\ = k-2u \leq k \text{ if } v\geq u.
    \end{cases}
    \]
\end{itemize}
Lastly,
\begin{align*}
h\br{h^{-1}(u,v)} &= \frac{1}{2}\begin{bmatrix}
    k+u+v-k+u-v\\
    k+u+v-k-u+v
\end{bmatrix}\\ 
&= \frac{1}{2}\begin{bmatrix}
  2u\\
  2v
\end{bmatrix} = (u,v),
\end{align*}
and
\begin{align*}
h^{-1}\br{h(x,y)} &= \begin{bmatrix}
    \frac{k+x+y}{2}+\frac{k+x-y}{2}-k\\
    \frac{k+x+y}{2}-\frac{k+x-y}{2}
\end{bmatrix}\\ &= \begin{bmatrix}
    k+x-k\\
    y
\end{bmatrix} = (x,y).
\end{align*}
Note that since $B$ is in bijection with $\Iintv{0,k}^2$, $|B| = |\Iintv{0,k}^2| = (k+1)^2$.
\subsubsection{Non-negative slopes and corresponding offsets}
From there,we find the set of \textit{non-negative} slopes,
\begin{equation*}
\cbr{\eps_1 (u+v-k) + \eps_2 (u-v) \ | \ \substack{ u,v \in \Iintv{0,k},\\ \ u \geq \left\lceil \frac{k\eps_1 - v(\eps_1 -\eps_2)}{\eps_1 + \eps_2}\right\rceil}}.
\end{equation*}
To see this, start from
\[
\text{Slopes} = \cbr{\eps_1(u+v-k) + \eps_2(u-v) \ | \ u,v \in \Iintv{0,k}}.
\]
We want to find all $u,v \in \Iintv{0,k}$ such that, rearranging,
\begin{equation*}
    \eps_1(u+v-k) + \eps_2(u-v) = u(\eps_1 +\eps_2) + v(\eps_1 - \eps_2) - k\eps_1 \geq 0.
\end{equation*}
Letting $v \in \Iintv{0,k}$, this imposes the following constraint on $u$:
\begin{equation*}
    u(\eps_1 + \eps_2) \geq k\eps_1 - v(\eps_1 -\eps_2) \iff u \geq \frac{k\eps_1 - v(\eps_1 -\eps_2)}{\eps_1 + \eps_2}.
\end{equation*}
To see this, observe that $\eps_1 \geq \eps_1 - \eps_2 $ hence $ v \eps_1 \geq v(\eps_1-\eps_2)$ and $k\eps_1 - v(\eps_1-\eps_2) \geq k \eps_1 - v\eps_1 = (k-v) \eps_1 \geq 0$ since $v \in \Iintv{0,k}$. Thus the constraint we found on $u$ is always active:
\begin{equation*}
    u \in \Iintv{\left\lceil \frac{k\eps_1 - v(\eps_1 -\eps_2)}{\eps_1 + \eps_2}\right\rceil, k}.
\end{equation*}
For each $\eps_{u,v} = \eps_1 (u+v-k) + \eps_2 (u-v) \geq 0$, we seek
\begin{equation*}
\delta_{u,v} = \max_{A \subseteq \Iintv{0,3}^k} \Tilde{P}_0(A) - e^{\eps_{u,v}}\Tilde{P}_1(A).
\end{equation*}
This can be written
\begin{equation*}
    \delta_{u,v} = \sum_{s^k \in \Iintv{0,3}^k} \max\br{0, \Tilde{P}_0(s^k) - e^{\eps_{u,v}}\Tilde{P}_1(s^k)}.
\end{equation*}
Observe, as in \cite{dptv}, that $\Tilde{P}_1(s^k)$ only depends on $s^k$ through the quantities $a,b,c,d$ defined above, i.e. the number of 0s, 1s, 2s and 3s, and not their placement. For a given tuple $(a,b,c,d) \in \Iintv{0,k}^4$ such that $a+b+c+d = k$, there are
\begin{equation*}
    \binom{k}{a} \binom{k-a}{b} \binom{k-(a+b)}{c} = \binom{k}{a,b,c,d}
\end{equation*}
sequences $s^k$ such that $a(s^k)=a, b(s^k) = b, c(s^k) = c$ and $d(s^k) = d.$  Hence, $\Tilde{P}_i(s^k) = \Tilde{P}_i(0^a1^b2^c3^d) := \Tilde{P}_i(s^{a,b,c,d})$.
\begin{align*}
    \delta_{u,v} = &\sum_{\substack{(a,b,c,d) \in \Iintv{0,k}^4 \\ a+b+c+d = k}} \binom{k}{a,b,c,d}\cdot\\&\max\br{0, \Tilde{P}_0(s^{a,b,c,d}) - e^{\eps_{u,v}}\Tilde{P}_1(s^{a,b,c,d})}.
\end{align*}
To simplify this computation, we seek to find $(a,b,c,d) \in \Iintv{0,k}^4$ with $a+b+c+d = k$ and
\begin{align*}
    &\Tilde{P}_0(0^a1^b2^c3^d) - e^{\eps_{u,v}}\Tilde{P}_1(0^a1^b2^c3^d)\\
    = &\br{\frac{1-\alpha}{e^{\eps_1}+1}}^{a+d} \br{\frac{\alpha}{e^{\eps_2} + 1}}^{b+c}\cdot\\ & \br{e^{a\eps_1 + b\eps_2} - e^{(d+u+v-k)\eps_1 + (c+u-v)\eps_2}} > 0
\end{align*}
or, equivalently,
\begin{equation*}
     (a+k-d-u-v)\eps_1 + (b+v-c-u)\eps_2 > 0.
\end{equation*}
Denote by $B(\eps_1,\eps_2;u,v)$ the set of solutions, i.e.
\begin{align*}
    &B(\eps_1,\eps_2;u,v) =\\ &\cbr{(a,b,c,d) \in \Iintv{0,k}^4 | \begin{cases} a+b+c+d = k, \\(a+k-d-u-v)\eps_1\\ + (b+v-c-u)\eps_2 > 0 \end{cases} }.
\end{align*}
This concludes the computation of the privacy region stated in Theorem \ref{thm:alt-main}. The converse proof follows a similar line as in section \ref{subsubsec:converse_proof}.

\subsection{Proof of Proposition \ref{prop:approx_below}}\label{subsec:proof_approx_below}
\noindent First observe that
\begin{align*}
    &\arg\min_{\substack{(\boldeps, \bolddelta) \in (\bbR^+)^2 \times [0,1]^2\\ 0 \leq  f_{\boldeps, \bolddelta} \leq f} } \int_0^1 f(t) - f_{\boldeps, \bolddelta}(t) \di t\\ &= \arg\max_{\substack{(\boldeps, \bolddelta) \in (\bbR^+)^2 \times [0,1]^2\\ 0 \leq  f_{\boldeps, \bolddelta} \leq f} } \int_0^1  f_{\boldeps, \bolddelta}(t) \di t
\end{align*}
because $\int_0^1 f(t) \di t$ is a constant w.r.t. $(\boldeps, \bolddelta)$. In the following, we will thus aim to maximize the integral $\int_0^1 f_{\boldeps, \bolddelta}(t)\di t$ under the constraint $f_{\boldeps, \bolddelta} \leq f$.\\
Consider the $\frac{\pi}{4}$ rotation matrix
\begin{equation*}
    R = \frac{1}{\sqrt{2}}\begin{bmatrix}
        1 & -1 \\
        1 & 1
    \end{bmatrix} = \begin{bmatrix}
        \cos\br{\frac{\pi}{4}} & -\sin\br{\frac{\pi}{4}}\\
        \sin\br{\frac{\pi}{4}} & \cos\br{\frac{\pi}{4}}
    \end{bmatrix}.
\end{equation*}
Rotate the graph of $f$
\begin{equation*}
    \graph{f} = \cbr{ (x,f(x)) \ | \ x \in [0,1]},
\end{equation*}
this yields
\begin{align*}
    R\br{ \graph{f}} &= \cbr{R(x,f(x)) \ | \ x \in [0,1]}\\  &= \cbr{\frac{1}{\sqrt{2}}\br{x-f(x), x+f(x)} \ | \ x \in [0,1]}.
\end{align*}
Let $u := s(x) = \frac{x-f(x)}{\sqrt{2}}$ for $x \in [0,1]$. Observe that $s$ is an increasing function of $x$: letting $1 \geq x_1 > x_2 \geq 0$, we have:
\begin{align*}
    \sqrt{2}(s(x_1) - s(x_2)) &= x_1 - f(x_1) - x_2 + f(x_2)\\
    &= \underbrace{x_1 - x_2}_{> 0 } + \underbrace{f(x_2) - f(x_1)}_{\substack{\geq 0 \text{ because } \\ f \text{ non-increasing}}}\\
    & > 0.
\end{align*}
Using the symmetry of $f$ and the fact that $s$ is increasing, we see that
\begin{equation*}
    \imag{s} = \sbr{\frac{-f(0)}{\sqrt{2}}, \frac{f(0)}{\sqrt{2}}} \subseteq \sbr{\pm \frac{1}{\sqrt{2}}}.
\end{equation*}
This shows that the following function is well-defined:
\begin{align*}
    g: \sbr{-\frac{f(0)}{\sqrt{2}},\frac{f(0)}{\sqrt{2}}} &\to \R,\\
    u &\mapsto g(u) = \frac{x+f(x)}{\sqrt{2}}, \text{ where } x = s^{-1}(u).
\end{align*}
Consequently,
\begin{equation*}
R(\graph{f}) = \graph{g}.
\end{equation*}
Observe that $g$ is the 45 degrees counter-clockwise rotated version of $f$. $g$ inherits the smoothness and convexity properties of $f$:
\begin{itemize}
    \item Smoothness (as smooth as $f$): because $g$ is obtained through inverting a linear transformation of $x$ and $f(x)$.
    \item Convexity: $\sqrt{2}g = (\cdot +f(\cdot )) \circ s^{-1}$ which is a composition of convex functions. Indeed, $s^{-1}$ is convex because $\sqrt{2}s''(x) = -f''(x) \leq 0$, hence $(s^{-1})'' \geq 0$. Note that when $f$ is strictly convex, $g$ is as well.
\end{itemize}
$f$'s self-symmetry property results in $g$ being even. Indeed, first observe
\begin{equation*}
    -u = \frac{f(x)-x}{\sqrt{2}} = \frac{f(x)-f(f(x))}{\sqrt{2}} = s(f(x)).
\end{equation*}
Hence
\begin{equation*}
    g(-u) = \frac{f(x)+f(f(x))}{\sqrt{2}} = \frac{f(x)+x}{\sqrt{2}} = g(u).
\end{equation*}
Thus, instead of approximating $f$ on $[0,1]$, we can equivalently approximate $g$ on $\sbr{-\frac{f(0)}{\sqrt{2}}, 0}$ by a 2-piece piecewise affine function $\tilde{g}$, and mirror $g$ with respect to $u=0$. Let $L_1$ and $L_2$ be the two affine pieces that approximate $g$, with $L_i(t) = \alpha_i t + \beta_i$. Let $L_1$ and $L_2$ meet at $t^*$, i.e. $L_1(t^*) = L_2(t^*)$, where $t^*$ is thus a function of the $\alpha_i,\beta_i$ parameters. Let $z = -\frac{f(0)}{\sqrt{2}}$. We seek to maximize
\begin{equation*}
    I = \int_z^{t^*} L_1(t)\di t + \int_{t^*}^0 L_2(t) \di t
\end{equation*}
under the constraint 
\begin{equation*}
    \max\cbr{L_1(t), L_2(t)} \leq g(t).
\end{equation*}
Observe that $L_i$ have to be tangent to $g$, by convexity of $g$. If $L_i$ is not tangent to $g$, its offset $\beta_i$ can be increased until $L_i$ becomes tangent to $g$, strictly increasing the value of the integral above. Thus, 
\begin{equation*}
    \begin{cases}
        L_1(t) = g(t_1) + g'(t_1)(t-t_1) \ \text{ for some } t_1 \in \sbr{z,t^*},\\
        L_2(t) = g(t_2) + g'(t_2)(t-t_2) \ \text{ for some } t_1 \in \sbr{t^*,0}.
    \end{cases}
\end{equation*}
$t^*$ can be explicitly written as a function of $t_1,t_2$ using $L_1(t^*) = L_2(t^*)$. Define
\begin{align*}
&J(t_1,t_2,t^*)\\ &= \int_z^{t^*} L_1(t)\di t + \int_{t^*}^0 L_2(t) \di t\\ &= \int_z^{t^*} g(t_1) + g'(t_1)(t-t_1)\di t + \int_{t^*}^0 g(t_2) + g'(t_2)(t-t_2) \di t
\end{align*}
and maximize $J$ under the constraint 
\begin{equation*}
    c(t_1,t_2,t^*) = L_1(t^*) - L_2(t^*) = 0.
\end{equation*}
To this end, consider the Lagrangian
\begin{equation*}
    L(t_1,t_2,t^*,\lambda) = J(t_1,t_2,t^*) - \lambda c(t_1,t_2,t^*).
\end{equation*}
Computing the partial derivatives of $L$, using the Leibnitz rule to differentiate under the integral:
\begin{align*}
    \frac{\partial L}{\partial t_1}(t_1,t_2,t^*,\lambda) &= \int_{z}^{t^*} g''(t_1)(t-t_1) \di t - \lambda g''(t_1)(t^* - t_1),\\
    \frac{\partial L}{\partial t_2}(t_1,t_2,t^*,\lambda) &= \int_{t^*}^{0} g''(t_2)(t-t_2) \di t + \lambda g''(t_2)(t^* - t_2),\\
    \frac{\partial L}{\partial t^*}(t_1,t_2,t^*,\lambda)&=L_1(t^*) - L_2(t^*) - \lambda \br{g'(t_1) - g'(t_2)},\\
    \frac{\partial L}{\partial \lambda}(t_1,t_2,t^*,\lambda) &= L_2(t^*) - L_1(t^*).
\end{align*}
Developing the expressions for $\frac{\partial L}{\partial t_1}$ and $\frac{\partial L}{\partial t_2}$:
\begin{align*}
    \int_{z}^{t^*} g''(t_1)(t-t_1) \di t &= \frac{g''(t_1)}{2}\sbr{(t^*-t_1)^2 - (z-t_1)^2}\\ &= g''(t_1)\br{t^*-z}\br{\frac{t^*+z}{2}-t_1},
\end{align*}
and similarly
\begin{equation*}
\int_{t^*}^{0} g''(t_2)(t-t_2) \di t = -g''(t_2)t^*\br{\frac{t^*}{2}-t_2^2}.
\end{equation*}
We set 
\begin{equation*}
    \nabla L = 0.
\end{equation*}
This yields the following system:
\begin{align*}
&\begin{cases}
    g''(t_1)\br{t^*-z}\br{\frac{t^*+z}{2}-t_1} = \lambda g''(t_1)(t^* - t_1)\\
    g''(t_2)t^*\br{\frac{t^*}{2}-t_2} = \lambda g''(t_2)(t^* - t_2)\\
    L_1(t^*) - L_2(t^*) = \lambda \br{g'(t_1) - g'(t_2)}\\
     L_1(t^*) - L_2(t^*) = 0
\end{cases}\\
\iff &\begin{cases}
    g''(t_1)\br{t^*-z}\br{\frac{t^*+z}{2}-t_1} = \lambda g''(t_1)(t^* - t_1)\\
    g''(t_2)t^*\br{\frac{t^*}{2}-t_2} = \lambda g''(t_2)(t^* - t_2)\\
    \lambda \br{g'(t_1) - g'(t_2)} = 0\\
    L_1(t^*) - L_2(t^*) = 0.
\end{cases}
\end{align*}
Assume strict convexity of $f$, implying strict convexity of $g$. Indeed, $g$ is only non strictly convex when $g$ is affine over an open interval, translating to an affine piece of $f$. Then $g'' \neq 0$ and the above simplifies to
\begin{align*}
&\begin{cases}
    \br{t^*-z}\br{\frac{t^*+z}{2}-t_1} = \lambda (t^* - t_1)\\
    t^*\br{\frac{t^*}{2}-t_2} = \lambda (t^* - t_2)\\
    \lambda \br{g'(t_1) - g'(t_2)} = 0\\
    L_1(t^*) - L_2(t^*) = 0
\end{cases}\\
\iff &\begin{cases}
    \lambda = \frac{\br{t^*-z}\br{\frac{t^*+z}{2}-t_1}}{t^* - t_1}\\
    \lambda = \frac{t^*\br{\frac{t^*}{2}-t_2}}{t^* - t_2}\\
    \lambda \br{g'(t_1) - g'(t_2)} = 0\\
    L_1(t^*) - L_2(t^*) = 0.
\end{cases}
\end{align*}
The third equation yields either $g'(t_1) - g'(t_2) = 0$, thus $L_1$ and $L_2$ are in fact the same line, or $\lambda = 0$ implying
\begin{equation*}
    t_1 = \frac{t^* + z}{2}, \ t_2 = \frac{t^*}{2}.
\end{equation*}
Observe that the first case is contained within the second one, as when $g'(t_1) = g'(t_2)$ both lines are the same and thus the constraint is met with $\lambda = 0$. Thus, w.l.o.g., the system becomes
\begin{align}
    &\begin{cases}
        t_1 = \frac{t^* + z}{2}\\
        t_2 = \frac{t^*}{2}\\
        L_1(t^*) = L_2(t^*)
    \end{cases}\nonumber \\\iff &\begin{cases}
        t_1 = \frac{t^* + z}{2}\\
        t_2 = \frac{t^*}{2}\\
        g(t_1) + g'(t_1)(t^*-t_1) = g(t_2) + g'(t_2)(t^*-t_2)
    \end{cases} \nonumber\\
    \iff &\begin{cases}
        t_1 = \frac{t^* + z}{2}\\
        t_2 = \frac{t^*}{2}\\
        g\br{\frac{t^*+z}{2}} + g'\br{\frac{t^*+z}{2}}\frac{t^*-z}{2} = g\br{\frac{t^*}{2}} + g'\br{\frac{t^*}{2}}\frac{t^*}{2}.
    \end{cases}\nonumber
\end{align}
Writing $L_i(t) = \alpha_it + \beta_i$, we have 
\begin{equation*}
    \alpha_i = g'(t_i), \ \beta_i = g(t_i) - g'(t_i)t_i.
\end{equation*}
Rotating back to the original orientation, we find:
\begin{equation*}
    R^{-1}(t, \alpha_i t + \beta_i) = (r,v) \iff v = \frac{\alpha_i-1}{\alpha_1 + 1}r + \frac{\sqrt{2}\beta_i}{\alpha_i+1}.
\end{equation*}
Thus, $\alpha_i t + \beta_i$ maps to the line $a_i t + b_i$ in the original orientation, where
\begin{equation*}
    a_i = \frac{\alpha_i-1}{\alpha_i + 1}, \ b_i = \frac{\sqrt{2}\beta_i}{\alpha_i+1}.
\end{equation*}
Equivalently, observing that $a_it + b_i = -e^{\eps_i}t +(1-\delta_i)$ for $t \in [0,c]$,
\begin{equation*}
    \eps_i = \ln\br{-\frac{\alpha_i-1}{\alpha_i + 1}}, \ \delta_i = 1-\frac{\sqrt{2}\beta_i}{\alpha_i+1}.
\end{equation*}
Observe that the case $\alpha_i = -1 \iff a_i = -\infty$ cannot occur under our smoothness conditions on $f$.

\subsection{Proof of Proposition \ref{prop:approx_above}}\label{subsec:proof_approx_above}
\noindent The first step of this proof is similar to the previous one.
\begin{align*}
    &\arg\min_{\substack{(\boldeps, \bolddelta) \in (\bbR^+)^2 \times [0,1]^2\\ 1 \geq f_{\boldeps, \bolddelta} \geq f} } \int_0^1 f_{\boldeps, \bolddelta}(t) - f(t) \di t\\ &= \arg\min_{\substack{(\boldeps, \bolddelta) \in (\bbR^+)^2 \times [0,1]^2\\ 1 \geq f_{\boldeps, \bolddelta} \geq f} } \int_0^1 f_{\boldeps, \bolddelta}(t)\di t.
\end{align*}
Thus we minimize the area under $f_{\boldeps,\bolddelta}$ with the constraint $f_{\boldeps, \bolddelta} \geq f$.\\
$f_{\boldeps,\bolddelta}$ is a four piece piecewise-affine function. $f_{\boldeps,\bolddelta}$ is also a trade-off function, thus $f_{\boldeps,\bolddelta} = f_{\boldeps,\bolddelta}^{-1}$ implying symmetry w.r.t. the line $y= x$. Above this line, $f_{\boldeps,\bolddelta}$ is represented by two affine pieces on $[0,c']$, $L_1(t) = a_1t + b_1 \ \text{ for } t \in [0,t^*]$ and $L_2(t) = a_2t + b_2 \ \text{ for } t \in [t^*,c']$ for some choice of $t^* \in [0,1]$, such that $L_2(c') = c'$. $L_1, L_2$ are then mirrored w.r.t. $y =x$, yielding $f_{\boldeps,\bolddelta}(t)$ for $t \in [c', 1]$. Observe that, for $f_{\boldeps,\bolddelta}$ to be optimal, $L_1(0) = f(0) := f_0$ and $c' = c \implies L_2(c) = f(c) = c$. Else, one can trivially obtain a better approximation, in this context a smaller integral, by lowering the lines until they reach these two points. Thus, we seek lines $L_1$ and $L_2$ such that
\begin{equation*}
    \begin{cases}
        L_1(t) = a_1t + b_1 \ \text{ for } t \in [0,t^*],\\
        L_2(t) = a_2t + b_2 \ \text{ for } t \in [t^*,c],\\
        L_1(0) = f(0) = f_0,\\
        L_2(c) = c,\\
        L_1(t^*) = L_2(t^*) = f(t^*)
    \end{cases}
\end{equation*}
and $\int_0^c \max\cbr{L_1(t),L_2(t)} \di t$ is minimal. We can simplify the system above to
\begin{equation*}
    \begin{cases}
        L_1(t) = \frac{f(t^*)-f_0}{t^*}t + f_0 \ &\text{ for } t \in [0,t^*],\\
        L_2(t) = \frac{f(t^*)-c}{t^*-c}\br{t-c}+c\ &\text{ for } t \in [t^*,c]
    \end{cases}
\end{equation*}
for some optimal choice $t^*$. Set 
\begin{equation*}
    J(t^*) = \int_{0}^{t^*} L_1(t) \di t + \int_{t^*}^c L_2(t) \di t.
\end{equation*}
Then, we compute each integral. In the following assume $t^* \not\in \cbr{0,c}$, else $J$ reduces to either of the 2 integrals.
\begin{align*}
 \int_{0}^{t^*} L_1(t) \di t &=  \int_0^{t^*}\frac{f(t^*)-f_0}{t^*}t + f_0 \di t\\
 &= f_0t^* + \frac{f(t^*)-f_0}{2}t^*\\
 &= t^* \frac{f(t^*)+f_0}{2},
\end{align*}
and
\begin{align*}
 \int_{t^*}^c L_2(t) \di t &= \int_{t^*}^c \frac{f(t^*)-c}{t^*-c}\br{t-c}+c\di t\\
 &= c(c-t^*) + \frac{1}{2}\frac{f(t^*)-c}{t^*-c}\sbr{(t-c)^2}_{t^*}^c\\
 &= c(c-t^*) + \frac{1}{2}\frac{f(t^*)-c}{t^*-c}\br{-(t^*-c)^2}\\
 &= c(c-t^*) - \frac{1}{2}\frac{f(t^*)-c}{t^*-c}(t^*-c)^2\\
 &= -(t^*-c)\sbr{c+\frac{f(t^*)-c}{2}}\\
 &= -(t^*-c)\frac{f(t^*)+c}{2}.
\end{align*}
The original problem becomes
\begin{align*}
&\arg\min_{\substack{(\boldeps, \bolddelta) \in (\bbR^+)^n \times [0,1]^n\\ f_{\boldeps, \bolddelta} \geq f} } \int_0^1 f_{\boldeps, \bolddelta}(t)  \di t \equiv\\ &\arg\min_{t^* \in [0,c]} J(t^*) = t^* \frac{f(t^*)+f_0}{2} -(t^*-c)\frac{f(t^*)+c}{2}.
\end{align*}
Computing derivatives $J'$ and $J''$ of $J$, which are well defined, because $f$ is twice differentiable:
\begin{align*}
J'(t^*) &= \frac{f(t^*)+f_0 + t^*f'(t^*)}{2} - \frac{f(t^*)+c}{2} - \frac{t^*-c}{2}f'(t^*)\\
&= \frac{f_0-c}{2}+\frac{c}{2}f'(t^*),\\
\text{and } J''(t^*) &= \frac{c}{2}f''(t^*) \geq 0.
\end{align*}
$J$ being convex, its stationary points are global minima. Observe that in the previous section, we have already required $f$ to be strictly convex. Hence $J$ is strictly convex as well ($J'' > 0$), and has a unique global minimum over $[0,c]$. We seek the stationary points of $J$:
\begin{align*}
    J'(t_\text{stat}) = 0 &\iff \frac{f_0-c}{2}+\frac{c}{2}f'(t_\text{stat}) = 0\\ &\iff f'(t_\text{stat}) = \frac{f(c) - f(0)}{c-0}.
\end{align*}
Observe that a solution $t_\text{stat} \in \br{0,c}$ to the previous equation exists by the mean value theorem, and is unique by the strict convexity of $J$. Solving that equation then comparing $J(t_\text{stat})$ to $J(0)$ and $J(c)$, we have that 
\begin{equation*}
     t^* := \arg\min_{t^* \in [0,c]} J(t^*)  =  \arg\min\cbr{J(0),J(c), J(t_\text{stat})}.
\end{equation*}
Observe that 
\begin{equation*}
    J(0) = J(c) = c\frac{c+f_0}{2}.
\end{equation*}
Hence $J(t_\text{stat}) < J(0) \iff J(t_\text{stat}) < J(c)$ and to decide between $t_\text{stat}$ and $\cbr{0,c}$, one can compute
\begin{equation*}
    h(t_\text{stat}) = 2(J(t_\text{stat})-J(0)) = t_\text{stat}f(0) + c(f(t_\text{stat})-f(0)-t_\text{stat}),
\end{equation*}
and we have
\begin{equation*}
   t_\text{stat} \in \arg\min\cbr{J(0),J(c), J(t_\text{stat})} \iff h(t_\text{stat}) \leq 0.
\end{equation*}
Lastly, we can compute the affine pieces. If $t^* \in \cbr{0,c} \iff h(t_\text{stat}) \geq 0$, then one of the two affine pieces collapses and only the second one remains, representing the $(\eps,\delta)$-DP constraint
\begin{equation*}
    \begin{cases}
        \eps = \ln\br{\frac{f(0)-c}{c}},\\
        \delta = 1-f(0).
    \end{cases}
\end{equation*}
When $t^* = t_\text{stat} \iff h(t_\text{stat}) \leq 0$,
\begin{equation*}
    \begin{aligned}
    &\begin{cases}
                \eps_1 = \ln\br{\frac{f(0)-f(t_\text{stat})}{t_\text{stat}}},\\
        \delta_1 = 1-f(0),\\
    \end{cases}\\
        &\begin{cases}
            \eps_2 = \ln\br{\frac{c-f(t_\text{stat})}{t_\text{stat}-c}},\\
            \delta_2 = 1-c-c\frac{c-f(t_\text{stat})}{t_\text{stat}-c}.
        \end{cases}
    \end{aligned}
\end{equation*}

\subsection{Computation of the mixture of piecewise affine functions}\label{subsec:computation_piecewise_affine}
Lemma \ref{lem:mixture} states that, given trade-off functions $\cbr{f_i}_{i=1}^n$ with associated weights $\cbr{\alpha_i}_{i=1}^n \subseteq (0,1]$ such that $\sum_{i=1}\alpha_i=1$, the trade-off function of the mixture of the corresponding $n$ hypothesis tests is
\begin{equation*}
    f_m(t) = \min_{\substack{t_i \in [0,1], \ i \in \Iintv{1,n}\\ \sum_{i=1}^n \alpha_i t_i = t}} \sum_{i=1}^n \alpha_i f_i(t_i).
\end{equation*}
This equation looks very similar to an infimal convolution. 
\begin{definition}[Infimal convolution]\label{def:infconv} 
Let $g_i: \R \to \R \cup\cbr{\pm\infty}$ be functions, for $i=1$ to $n$. Their infimal convolution $g: \R \to \R \cup\cbr{\pm\infty}$ is defined by
\begin{equation*}
    g(t) = \br{\square_{i=1}^n g_i}(t) = \inf_{\br{x_i}_{i=1}^n: \sum_{i=1}^n x_i = t} \sum_{i=1}^n g_i(x_i).
\end{equation*}
\end{definition}
$f_m$ can be written as an infimal convolution of some functions. Extend $f_i$ by letting $f_i(t) = +\infty$ when $t \not \in [0,1]$. Define
\begin{equation*}
    g_i(x) = \alpha_i f_i\br{\frac{x}{\alpha_i}},
\end{equation*}
where $g_i(x) = +\infty \iff x \not\in [0,\alpha_i]$. Effectively, this is a change of variables $x_i = t_i \alpha_i$. We can write
\begin{align*}
    f_m(t) &= \min_{\substack{t_i \in [0,1], \ i \in \Iintv{1,n}\\ \sum_{i=1}^n \alpha_i t_i = t}} \sum_{i=1}^n \alpha_i f_i(t_i)\\
    &= \min_{\substack{t_i \in [0,1], \ i \in \Iintv{1,n}\\ \sum_{i=1}^n \alpha_i t_i = t}} \sum_{i=1}^n g_i(\alpha_it_i)\\
    &=\min_{\substack{\alpha_i t_i \in [0,\alpha_i], \ i \in \Iintv{1,n}\\ \sum_{i=1}^n \alpha_i t_i = t}} \sum_{i=1}^n g_i(\alpha_it_i)\\
    &= \min_{\substack{x_i \in [0,\alpha_i], \ i \in \Iintv{1,n}\\ \sum_{i=1}^n x_i = t}} \sum_{i=1}^n g_i(x_i)\\
    &= \min_{\substack{x_i \in \R, \ i \in \Iintv{1,n}\\ \sum_{i=1}^n x_i = t}} \sum_{i=1}^n g_i(x_i)\\
    f_m(t) &= \br{\square_{i=1}^n g_i}(t).
\end{align*}
The second to last step is due to $g_i(x) = +\infty$ for $x \notin [0,\alpha_i]$.

An infimal convolution can be computed through the \textbf{convex conjugate} \cite{boyd_convex_optimization}.
\begin{definition}[Convex conjugate]
Let $f: \R \to \R\cup \cbr{\pm\infty}$ be a function. Its convex conjugate or Legendre-Fenchel conjugate $f^*:\R \to \R\cup \cbr{\pm\infty}$ is defined by
\begin{equation*}
    f^*(s) = \sup_{t\in\R} st - f(t).
\end{equation*}
\end{definition}
Known convex conjugate properties \cite{boyd_convex_optimization} include the sequel.
\begin{proposition}[Convex conjugate of an infimal convolution]
A few properties of the convex conjugate will be important to us. They are based on the fact that a proper function $f = f^{**}$ if and only if $f$ is convex and lower semi-continuous.\\
Let $g_i$ be proper, convex and lower semicontinuous. Then
\begin{enumerate}
    \item the convex conjugate of the infimal convolution of $g_i$ is the sum of their convex conjugates:
    \begin{equation*}
        \br{\square_{i=1}^n g_i}^* = \sum_{i=1}^n g_i^*.
    \end{equation*} 
    \item the convex conjugate of their summation is the infimal convolution of their convex conjugates: \begin{equation*}
        \br{\sum_{i=1}^n g_i}^* = \square_{i=1}^n g_i^*.
    \end{equation*}
\end{enumerate}
Lastly, observe that the infimal convolution of proper, convex and lower semicontinuous functions is convex and lower semicontinuous.
\end{proposition}
Coupling these statements, we obtain the following lemma.
\begin{lemma}
\begin{equation*}
    f_m = f_m^{**} = \br{\sum_{i=1}^n g_i^*}^* = \br{\sum_{i=1}^n \alpha_i f_i^*}^*.
\end{equation*}
\end{lemma}
\begin{proof}
To prove this, it remains to see that $g_i^* = \alpha_i f_i^*$:
\begin{align*}
    g_i^*(s) &= \sup_{t\in \R} st - g_i(t)\\
    &= \sup_{t\in \R} st - \alpha_i f_i\br{\frac{t}{\alpha_i}}\\
    &= \sup_{t\in \R} s\alpha_i\frac{t}{\alpha_i} - \alpha_i f_i\br{\frac{t}{\alpha_i}}\\
    &= \alpha_i\sup_{\frac{t}{\alpha_i}\in \R} s\frac{t}{\alpha_i} -  f_i\br{\frac{t}{\alpha_i}}\\
    &= \alpha_i\sup_{u\in \R} su -  f_i\br{u}\\
    g_i^*(s)&= \alpha_i f_i^*(s).
\end{align*}
\end{proof}
The formula above is only interesting if $f_i^*$ and the conjugate of their convex combinations can be computed. A particularly important case arises when $f_i$ is piecewise affine, which is the case when $f_i$ is the trade-off function of an intersection of $(\eps,\delta)$-DP constraints.

Let $f_i(t) = \max\cbr{a_{ij}t + b_{ij}}_{j=1}^{m_i}$ for $t \in [0,1]$, and $f_i(t) = +\infty$ else. We can compute $f_i^*$ through the following lemma.
\begin{proposition}[Convex conjugate of a piecewise affine function with bounded domain]\label{prop:conj_bounded}
Let $f(t) = \max\cbr{a_{j}t + b_{j}}_{j=1}^{m}$ for $t \in [t_0, t_m]$ and $+\infty$ else, such that each line constraint $(a_j,b_j)$ is active at least at one $t$, and $a_j < a_{j+1}$ for $j \in [1,m-1]$. Then, $\forall s \in \R$,
\begin{equation*}
f^*(s) = \max \cbr{t_j s - f(t_j)}_{i=0}^m
\end{equation*}
where $t_j = -\frac{b_{j+1}-b_j}{a_{j+1}-a_j} $ for $j \in [1,m-1]$.
\end{proposition}
The computation follows from \cite{boyd_convex_optimization}. Observe that $f_i^*$ is thus again a piecewise affine function, defined over the unbounded domain $\R$. $\sum_{i=1}^n \alpha_i f_i^*$ is thus itself a piecewise affine proper function defined over $\R$. The next proposition tells us how to compute its convex conjugate, which is similar - also a result from \cite{boyd_convex_optimization}.
\begin{proposition}[Convex conjugate of a piecewise affine function with unbounded domain]\label{prop:conj_unbounded}
Let $f(t) = \max\cbr{a_{j}t + b_{j}}_{j=1}^{m}$ for $t \in \R$, such that each line constraint $(a_j,b_j)$ is active at least at one $t$, and $a_j < a_{j+1}$ for $j \in [1,m-1]$. Then, $\forall s \in [a_1,a_m]$,
\begin{equation*}
f^*(s) = \max \cbr{t_j s - f(t_j)}_{i=0}^m 
\end{equation*}
where $t_j = -\frac{b_{j+1}-b_j}{a_{j+1}-a_j} $ for $j \in [1,m-1]$, and $f^*(s) = +\infty$ for $s \not\in [a_1,a_m]$.
\end{proposition}
These two propositions suggest the following algorithm to compute $f_m = \br{\sum_{i=1}^n \alpha_i f_i^*}^*$:
\begin{itemize}
    \item Compute $f_i^*$ using Proposition \ref{prop:conj_bounded}.
    \item Find the slopes and offsets of the piecewise affine function $f_m^* = \sum_{i=1}^n \alpha_i f_i^*$.
    \item Finally, conjugate back $f_m^*$ to find $f_m = f_m^{**}$ using Proposition \ref{prop:conj_unbounded}.
\end{itemize}

\else
   \IEEEtriggeratref{14}
\fi
%%%%%%%%%%%%%%%%%%%%%%%%%%%%%%%

\bibliographystyle{IEEEtran}
\bibliography{References}

\end{document}